\documentclass[11pt]{article}

\usepackage[ruled]{algorithm2e} 

\usepackage{braket}
\usepackage{amsfonts}
\usepackage{amssymb}
\usepackage{amsmath}
\usepackage{latexsym}
\usepackage{amsthm}
\usepackage[usenames,dvipsnames]{color}
\usepackage{soul}
\usepackage{hyperref}
\usepackage{paralist}
\usepackage{fullpage}

\hypersetup{pdfpagemode=UseNone}

\newtheorem{theorem}{Theorem}[section]

\newtheorem{lemma}[theorem]{Lemma}

\newtheorem{definition}[theorem]{Definition}

\newenvironment{mylist}[1]{\begin{list}{}{
	\setlength{\leftmargin}{#1}
	\setlength{\rightmargin}{0mm}
	\setlength{\labelsep}{2mm}
	\setlength{\labelwidth}{8mm}
	\setlength{\itemsep}{0mm}}}
	{\end{list}}

\newcommand{\tinyspace}{\mspace{1mu}}

\newcommand{\abs}[1]{\left\lvert\tinyspace #1 \tinyspace\right\rvert}

\newcommand{\norm}[1]{\left\lVert\tinyspace #1 \tinyspace\right\rVert}
\newcommand{\inlinenorm}[1]{\lVert\tinyspace #1 \tinyspace\rVert}
\newcommand{\tr}{\operatorname{Tr}}

\newcommand{\setft}[1]{\mathrm{#1}}
\newcommand{\lin}[1]{\setft{L}\left(#1\right)}

\newcommand{\unitary}[1]{\setft{U}\left(#1\right)}
\newcommand{\herm}[1]{\setft{Herm}\left(#1\right)}

\newcommand{\class}[1]{\textup{#1}}
\newcommand{\prob}[1]{\textsc{#1}}
\newcommand{\ayes}{A_{\rm yes}} 
\newcommand{\ano}{A_{\rm no}} 

\newcommand{\inner}[2]{\langle #1, #2 \rangle}

\newcommand{\ketbra}[2]{\ket{#1}\!\bra{#2}}
\newcommand{\kb}[1]{\ketbra{#1}{#1}}
\newcommand{\brakett}[2]{\mbox{$\langle #1  | #2 \rangle$}}

\newcommand{\enorm}[1]{\norm{#1}_{\mathrm{2}}}      
\newcommand{\trnorm}[1]{\norm{#1}_{\mathrm {tr}}}  
\newcommand{\fnorm}[1]{\norm{#1}_{\mathrm {F}}}    
\newcommand{\snorm}[1]{\norm{#1}_{\mathrm {\infty}}}    
\newcommand{\inlinesnorm}[1]{\inlinenorm{#1}_{\mathrm {\infty}}}
\newcommand{\mnorm}[1]{\norm{#1}_{\max}}    

\newcommand{\gscon}{\prob{GSCON}}
\newcommand{\ffgscon}{\prob{FF-GSCON}}
\newcommand{\succgscon}{\prob{SUCCINCT GSCON}}
\newcommand{\osat}{\prob{ORACLE $3$SAT}}
\newcommand{\stcon}{\prob{s,t-CONN}}
\newcommand{\spa}[1]{\mathcal{#1}}
\newcommand{\pacc}{p_{\rm accept}}
\newcommand{\prej}{p_{\rm reject}}

\newcommand{\trace}{\tr}

\newcommand{\histstate}{\ket{\psi_{\rm hist}}}
\newcommand{\histstateketbra}{\ketbra{\psi_{\rm hist}}{\psi_{\rm hist}}}

\newcommand{\enc}[1]{\left<#1\right>}

\newcommand{\travlemmatwo}{Traversal Lemma}

\newcommand{\aaa}{\eta_1}
\newcommand{\bbb}{\eta_2}
\newcommand{\ccc}{\eta_3}
\newcommand{\ddd}{\eta_4}
\newcommand{\round}{\operatorname{round}}

\newcommand{\poly}{\textup{poly}}

\newcommand{\klh}{$k$-LH}

\def\complex{\mathbb{C}}
\def\real{\mathbb{R}}
\def\natural{\mathbb{N}}

\def\({\left(}
\def\){\right)}

\newcommand{\ve}[1]{\mathbf{#1}}

\usepackage{graphicx}
\newcommand{\red}[1]{#1}

\begin{document}

\title{Ground state connectivity of local Hamiltonians\footnote{A preliminary version of this paper appeared in Proceedings of the 42nd International Colloquium on Automata, Languages and Programming (ICALP), volume 9134 of Lecture Notes in Computer Science, pages 617 -- 628, 2015. The full version is published in ACM Transactions on Computation Theory (TOCT), Volume 10, Issue 2, May 2018, Article No. 8.}
}

\author{
  Sevag Gharibian\footnote{Simons Institute for the Theory of Computing and Department of Electrical Engineering and Computer Sciences, University of California, Berkeley, CA 94720, USA, and Department of Computer Science, Virginia Commonwealth University, Richmond, VA 23284, USA. Email: \tt{sevag.gharibian@gmail.com}.}
  \and
  Jamie Sikora\footnote{Centre for Quantum Technologies and MajuLab, CNRS-UNS-NUS-NTU International Joint Research Unit, UMI 3654, National University of Singapore, Singapore 117543. Email: \tt{jamiesikora@gmail.com}.}
}

\date{\today}

%


\maketitle
\begin{abstract}
The study of ground state energies of local Hamiltonians has played a fundamental role in quantum complexity theory. In this paper, we take a new direction by introducing the physically motivated notion of ``ground state connectivity'' of local Hamiltonians, which captures problems in areas ranging from quantum stabilizer codes to quantum memories. {Roughly, ``ground state connectivity'' corresponds to the natural question: Given two ground states $\ket{\psi}$ and $\ket{\phi}$ of a local Hamiltonian $H$, is there an ``energy barrier'' (with respect to $H$) along any sequence of local operations mapping $\ket{\psi}$ to $\ket{\phi}$? We show that the complexity of this question} can range from \class{QCMA}-complete to \class{PSPACE}-complete, as well as \class{NEXP}-complete for an appropriately defined ``succinct'' version of the problem. As a result, we obtain a natural \class{QCMA}-complete problem, a goal which has generally proven difficult since the conception of \class{QCMA} over a decade ago. Our proofs rely on a new technical tool, the Traversal Lemma, which analyzes the Hilbert space a local unitary evolution must traverse under certain conditions.  We show that this lemma is {essentially} tight
with respect to the length of the unitary evolution in question.
\end{abstract}

\section{Introduction}\label{sect:intro}

Over the last fifteen years, the merging of condensed matter physics and computational complexity theory has given rise {to} a new field of study known as \emph{quantum Hamiltonian complexity}~\cite{O11,GHLS14}. The cornerstone of this field is arguably Kitaev's~\cite{KSV02} quantum version of the Cook-Levin theorem~\cite{C72,L73}, which says that the problem of estimating the ground state energy of a local Hamiltonian is complete for the class Quantum Merlin Arthur (\class{QMA}), where \class{QMA} is a natural generalization of NP. Here, a $k$-local Hamiltonian is an operator $H=\sum_i H_i$ acting on $n$ qubits, such that each local Hermitian constraint $H_i$ acts non-trivially on $k$ qubits. The \emph{ground state energy} of $H$ is simply the smallest eigenvalue of $H$, and the corresponding eigenspace is known as the \emph{ground space} of $H$.

 Kitaev's result spurred a long line of subsequent works on variants of the ground energy estimation problem (see, e.g.~\cite{O11,GHLS14} for surveys), known as the $k$-local Hamiltonian problem (\klh). For example, Oliveira and Terhal showed that LH remains \class{QMA}-complete in the physically motivated case of qubits arranged on a 2D lattice~\cite{OT05}. Bravyi and Vyalyi proved~\cite{BV05} that the \emph{commuting} variant of $2$-LH is in NP. More recently, the complexity of the version of $2$-LH in which large positive and negative weights on local terms are allowed\footnote{Note that certain physically motivated local Hamiltonian models, such as the Heisenberg anti-ferromagnet (see, e.g.,~\cite{GHLS14} for a definition), require unit weights on all constraints, and are thus not captured by the dichotomy theorem of~\cite{CM13}.} was characterized by Cubitt and Montanaro~\cite{CM13} in a manner analogous to Schaeffer's dichotomy theorem for Boolean satisfiability~\cite{S78}. Thus, $k$-LH has served as an excellent ``benchmark'' problem for delving into the complexity of problems encountered in the study of local Hamiltonians. Yet, one can also ask about the \emph{properties of the ground space} itself. For example, is it topologically ordered? Can we evaluate local observables against it (e.g. for non-degenerate ground state $\ket{\psi}$ and $2$-local observable $O$, can one estimate $\bra{\psi}I\otimes O\ket{\psi}$)? It is this direction which we pursue in this paper.

Specifically, in this paper we define a notion of \emph{connectivity} of the ground space of $H$, which roughly asks: Given ground states $\ket{\psi}$ and $\ket{\phi}$ of $H$ as input, are they ``connected'' through the ground space of $H$? Somewhat more formally, we have (see Section~\ref{sect:defns} for a formal definition):

\begin{definition}[Ground State Connectivity (\gscon) (informal)]\label{def:GSCONinformal}
Given as input a local Hamiltonian $H$ and two ground states $\ket{\psi}$ and $\ket{\phi}$ (represented succinctly via quantum circuits) of $H$, as well as parameters $m$ and $l$, does there exist a sequence of $l$-qubit unitaries $\left(U_i\right)_{i=1}^m$ such that:
\begin{enumerate}
    \item ($\ket{\psi}$ mapped to $\ket{\phi}$) $U_m\cdots U_1 \ket{\psi}\approx \ket{\phi}$, and
    \item (intermediate states have low energy) $\forall i\in [m]$, $U_i\cdots U_1\ket{\psi}$
    {has low energy with respect to $H$.}
\end{enumerate}
\end{definition}
\noindent In other words, \gscon~asks whether there exists a sequence of $m$ unitaries, each acting on (at most) $l$ qubits, mapping the initial state $\ket{\psi}$ to the final state $\ket{\phi}$ ``{through}'' the ground space of $H$. We stress that the parameters $m$ (i.e. number of unitaries) and $l$ (i.e. the locality of each unitary) are key; as we discuss shortly, depending on their setting, the complexity of \gscon~can vary greatly. (Note: While the most general formulation of \gscon\ above does not require intermediate states to lie exactly in the ground space of $H$, our QCMA-completeness result holds even if one requires all intermediate states to lie fully in the ground space (see Section~\ref{sect:GSCON_QCMA}).)


\paragraph{Physics Motivation.} The original inspiration for this work came from a recently active area in classical complexity theory on \emph{reconfiguration} problems (see \emph{Previous work} below for details). For example, the reconfiguration problem for 3SAT asks: Given a 3SAT formula $\phi$ and satisfying assignments $x$ and $y$ for $\phi$, does there exist a sequence of bit flips mapping $x$ to $y$, such that each intermediate assignment encountered is also a satisfying assignment for $\phi$?   Although the classical study of reconfiguration problems is arguably mostly interesting from a theoretical perspective (i.e. it is theoretically interesting to ask about the structure of the solution space of a 3SAT instance, but we are not aware of any practical applications), its quantum variant (i.e. $\gscon$) turns out to be physically relevant. {In particular, it corresponds to the question: Given two ground states $\ket{\psi}$ and $\ket{\phi}$ of a local Hamiltonian $H$, are $\ket{\psi}$ and $\ket{\phi}$ separated by an ``energy barrier'' (with respect to $H$ and sequences of local unitaries mapping $\ket{\psi}$ to $\ket{\phi}$)? Along these lines, we now {discuss} connections to \emph{quantum memories} and \emph{stabilizer codes}.}\\

\noindent {\emph{Quantum memories.} A key challenge in building quantum computers is the implementation of long-lived qubit systems. In low-temperature systems, one approach is to encode a qubit in the ground state of a gapped Hamiltonian with a degenerate ground space. Here, the degeneracy ensures {the} {ground space} has at least two basis states, logical $\ket{\widetilde{0}}$ and $\ket{\widetilde{1}}$, and the gap ensures that external noise does not (easily) take a ground state out of the ground space. However, this is not sufficient --- although environmental noise may not take {the state \emph{out} of the ground} space, it can still alter {the} state \emph{within} the {ground} space (e.g. inadvertently map $\ket{\widetilde{0}}$ to $\ket{\widetilde{1}}$). Thus, making the typical assumption that errors act locally, it should ideally not be possible for $\ket{\widetilde{0}}$ to be mapped to $\ket{\widetilde{1}}$ through the ground space via a sequence of local operations.} This is precisely the principle behind Kitaev's toy chain model~\cite{Kit01}, and the motivation behind the toric code~\cite{Kit03} (see also~\cite{KL09}). This notion of how ``robust'' a quantum memory is can thus be phrased as an instance of \gscon: Given a gapped Hamiltonian $H$, a ground state $\ket{\psi}$ {to which the} quantum memory is initialized, and an undesired ground state $\ket{\phi}$, is there a sequence of local errors mapping the state of our quantum memory through the ground space from $\ket{\psi}$ to $\ket{\phi}$?\\

\noindent {\emph{Stabilizer codes.} Roughly, a stabilizer code~\cite{G97} is a quantum error-correcting code defined by a set of commuting Hermitian operators, $S=\set{G_1, \ldots, G_k}$, such that $G_i\neq -I$ and $\snorm{G_i}\leq 1$ for all $G_i\in S$. The \emph{codespace} for $S$ is the set of all $\ket{\psi}$ satisfying $G_i \ket{\psi} = \ket{\psi}$ {for all $i \in [k]$}. In other words, defining $G_i^+$ as the projection onto the $+1$ eigenspace of $G_i$, the codespace is the ground space of the positive semidefinite Hamiltonian $H := \sum_{i=1}^k (I - G_i^+)$.
Typically, errors are assumed to occur on a small number of qubits at a time; with this assumption in place, the following is a {special case} of \gscon: Given $H$ and codewords $\ket{\psi}$ and $\ket{\phi}$, does there exist a sequence of at most $m$ local errors mapping $\ket{\psi}$ to $\ket{\phi}$, such that the entire error process is undetectable, i.e. each intermediate state remains in the codespace? (We leave the issue of how deep the connection between \gscon\ and stabilizer codes runs open. In particular, a nice question is whether GSCON for stabilizer codes can be solved \emph{efficiently}, i.e. in P. For comparison, solving for ground states of stabilizer code Hamiltonians is indeed in P~\cite{YB12}, whereas estimating ground state energies of general local Hamiltonians is QMA-complete~\cite{KSV02}.)


\paragraph{Results.} Having motivated \gscon, we now informally state our results.

\begin{theorem}[See Theorem~\ref{thm:qcmacomplete} for a formal statement]\label{thm:informalqcmacomplete}
    \gscon~for polynomially large $m$ (i.e. for polynomially many local unitaries $U$) and $l=2$ (i.e. $2$-qubit unitaries) is \class{QCMA}-complete.
\end{theorem}

\noindent Here, \class{QCMA} is \class{QMA} except with a classical {witness}~\cite{AN02}. See Section~\ref{sect:defns} for a formal definition. Theorem~\ref{thm:informalqcmacomplete} says that determining whether there exists a polynomial-size quantum circuit mapping $\ket{\psi}$ to $\ket{\phi}$ through the ground space of $H$ is \class{QCMA}-complete.

\begin{theorem}[See Theorem~\ref{thm:pspacecomplete} for a formal statement]\label{thm:informalpspacecomplete}
    \gscon~for exponentially large $m$ (i.e. for exponentially many local unitaries $U$) and $l=1$ (i.e. $1$-qubit unitaries) is \class{PSPACE}-complete.
\end{theorem}

\noindent Theorem~\ref{thm:informalpspacecomplete} says that determining whether there exists an exponential length sequence of {$1$-qubit unitaries} mapping $\ket{\psi}$ to $\ket{\phi}$ through the ground space of $H$ is \class{PSPACE}-complete.

Finally, in Section~\ref{sect:GSCON_nexp} we define a succinct variant of \gscon, called \succgscon, in which the Hamiltonian $H$ has a succinct circuit description, and the initial and final states $\ket{\psi}$ and $\ket{\phi}$ are product states. We show:

\begin{theorem}[See Theorem~\ref{thm:nexpcomplete} for a formal statement]\label{thm:informalnexpcomplete}
    \succgscon~for exponentially large $m$ (i.e. for exponentially many local unitaries $U$) and $l=1$ (i.e. $1$-qubit unitaries) is \class{NEXP}-complete.
\end{theorem}
\noindent As Theorem~\ref{thm:informalnexpcomplete} follows from techniques similar to Theorems~\ref{thm:informalqcmacomplete} and~\ref{thm:informalpspacecomplete}, we give only a proof sketch of it in Section~\ref{sect:GSCON_nexp}.

We remark that the choices of $m$ and $l$ above are key to our results. For example, Theorem~\ref{thm:informalqcmacomplete} holds for any constant $l\geq 2$ (see remarks after its proof); however, for $l\in \omega(\log N)$ (for $N$ the input size) the problem is likely no longer in QCMA, as the prover cannot send a classical description of each local unitary. Similarly, attempting to extend Theorem~\ref{thm:informalpspacecomplete} by setting $l=2$ appears problematic, as then any intermediate state in the unitary evolution seems to require exponential space to represent. This {modified} problem (i.e. Theorem~\ref{thm:informalpspacecomplete} with $l=2$) is, however, in NEXP, {and we conjecture} it to be NEXP-complete.


\paragraph{Proof techniques.} Our results rely on a new technical lemma called the Traversal Lemma, as well as the use of $\epsilon$-nets and  $\epsilon$-pseudo-nets (also known as \emph{improper covering sets}). We now outline the proof techniques behind Theorem~\ref{thm:qcmacomplete} (QCMA-completeness) in more detail; using similar ideas, Theorems~\ref{thm:pspacecomplete} (PSPACE-completeness) and~\ref{thm:nexpcomplete} (NEXP-completeness) follow analogously.

{Specifically, we outline both} \class{QCMA}-hardness and containment in \class{QCMA}. Beginning with the former, the central idea behind the construction is as follows. Let $V$ be an arbitrary \class{QCMA} verification circuit, and let $H'$ be the local Hamiltonian obtained from $V$ via Kitaev's circuit-to-Hamiltonian construction~\cite{KSV02} (see Lemma~\ref{l:kemperegev} for Kempe and Regev's $3$-local version~\cite{KR03}). Then, we design the input Hamiltonian $H$ to \gscon~so that ``traversing its ground space'' is equivalent to simulating the following protocol (i.e. an honest prover acts as follows): Suppose $H'$ acts on register $h$. Add three additional ancilla qubits (which we call $GO$ qubits), and prepare initial state $\ket{\psi}=\ket{0\cdots 0}_h\ket{000}_G$. Now, using two-qubit unitaries, prepare the ground state of $H'$ in register $h$ {(which can be done efficiently since $V$ is a QCMA circuit)}. Then, flip the three GO qubits using local Pauli $X$ gates to obtain $\ket{111}$ in $G$, and uncompute the history state in $h$ to obtain target state $\ket{\phi}=\ket{0\cdots 0}_h\otimes \ket{111}_G$. To enforce this honest behavior, we use $5$-local Hamiltonian $H$:
\begin{equation}
H := H'_h \otimes P_G \quad \text{for} \quad P := I - \kb{000} - \kb{111}.
\end{equation}
Note that the initial and final states $\ket{\psi}$ and $\ket{\phi}$ lie in the null space of $H$, and flipping a GO qubit ``activates'' the check Hamiltonian $H'_h$ which checks if $h$ has a valid and accepting history state. The pressing question is whether for a NO input, a cheating prover can somehow deviate from this protocol by flipping all three GO qubits using $2$-qubit unitaries \emph{without} ``activating'' $H'$. To rigorously show this is impossible, we state and prove our main technical tool, the Traversal Lemma {(Lemma~\ref{l:traversal})}, which roughly says that to transition from $\ket{000}$ to $\ket{111}$ in $G$ using $2$-qubit unitaries, an intermediate state in the evolution must have high overlap with $P_G$.

Let us elaborate further on the Traversal Lemma, which analyzes the Hilbert space a local unitary evolution must traverse in certain settings. Specifically, define two states $\ket{\psi}$ and $\ket{\phi}$ as \emph{$k$-orthogonal} if for any $k$-local unitary $U$, we have $\bra{\phi}U\ket{\psi}=0$. In other words, any application of a $k$-local unitary leaves $\ket{\psi}$ and $\ket{\phi}$ orthogonal. Then, the Traversal Lemma {roughly} says that for $k$-orthogonal states $\ket{\psi}$ and $\ket{\phi}$, if we wish to map $\ket{\psi}$ to $\ket{\phi}$ via a sequence of $k$-local unitaries, then at some step in this evolution we must leave the space spanned by $\ket{\psi}$ and $\ket{\phi}$, i.e. we must have {``large'' inner product with $I - \kb{\psi} - \kb{\phi}$}. (Here, ``large'' means the inner product scales at least as $\Omega(1/m^2)$, for $m$ the number of $k$-local unitaries applied.)
To prove the Traversal Lemma, we use a combination of the Gentle Measurement Lemma of Winter~\cite{W99} and an idea inspired by the quantum Zeno effect.

As the Traversal Lemma is a key technical contribution of this paper, we also study its properties further (i.e. independently of its application to our complexity theoretic results). For example, we show the lemma is tight up to a polynomial factor in the number of unitaries, $m$. To do so, we give a pair of ${2}$-orthogonal states $\ket{\psi}$, $\ket{\phi}$ with the following property: For any $0<\Delta<1/2$, we construct a carefully selected sequence of $O(1/\Delta^2)$ $2$-local unitaries mapping $\ket{\psi}$ to $\ket{\phi}$, such that at any point in this mapping, the
{inner product with $I - \kb{\psi} - \kb{\phi}$}
is at most $\Delta$. We also delve further into the study of $k$-orthogonality, including giving an intuitive characterization of the notion.

Finally, containment of \gscon~in \class{QCMA} is shown via a simple and natural verification procedure, wherein the prover sends a classical description of the local unitaries $\set{U_i}$, and the verifier prepares many copies of the starting, final, and all intermediate states and checks that all required properties hold. To make this rigorous\footnote{For clarity, $\epsilon$-pseudo-nets are used to avoid precision issues for unitaries containing irrational numbers. Alternatively, one could consider fixing a universal gate set, which unlike $\epsilon$-pseudo-nets, would make our QCMA containment result (Lemma~\ref{l:inQCMA}) dependent on the choice of gate set. This is perhaps not ideal, as the number of unitaries $m$ in Lemma~\ref{l:inQCMA} is polynomial, like the overhead required to switch from one universal gate set to another.}, we construct an \emph{$\epsilon$-pseudo-net}, which allows us to easily discretize the space of $d$-dimensional unitary operators for any $d\geq 2$. Such pseudo-nets come with a tradeoff: On the negative side, they contain non-unitary operators. On the positive side, they are not only straightforward to construct, but more importantly, they have the following property: Given any element $A$ in the pseudo-net, there are efficient \emph{explicit} protocols for checking if $A$ is close to unitary, and if so, for ``rounding'' it to such a unitary.


\paragraph{Previous work.} To the best of our knowledge, our work is the first to study reconfiguration in the quantum setting. In contrast, in the classical setting, such problems have recently received much attention. In particular, our work was inspired by the paper of Gopalan, Kolaitis, Maneva, and Papadimitriou~\cite{GKMP06}, which shows that determining whether two solutions $x$ and $y$ of a Boolean formula are connected through the solution space is either in P or is \class{PSPACE}-complete, depending on the constraint types allowed in the formula. (Note: A minor error in Reference~\cite{GKMP06} was recently corrected in the work of Schwerdtfeger~\cite{S13}.) More recently, Mouawad, Nishimura, Pathak and Raman~\cite{MNPR14} studied the variant of this problem in which one seeks the \emph{shortest} possible Boolean reconfiguration path; they show this problem is either in P, NP-complete, or \class{PSPACE}-complete. In this sense, our definition of \gscon~can be thought of as a quantum generalization of the problem studied in Reference~\cite{MNPR14}. More generally, since the work of Reference~\cite{GKMP06}, a flurry of papers have appeared studying reconfiguration for problems ranging from Boolean satisfiability to vertex cover to graph coloring~\cite{CHJ08,BC09, BJLPP11,CHJ11,FHHH11,IDHPSUU11,B12,IKD12,IKOZ12,KMM12,S13,BB13,MNRSS13,MNPR14,MNR14}.


\paragraph{Significance to complexity theory.} We now discuss the motivation behind \gscon~from a complexity theoretic perspective. We begin by focusing on $\class{QCMA}$, which is a natural class satisfying $\class{MA}\subseteq \class{QCMA}\subseteq\class{QMA}$. Although $\class{QCMA}$ was introduced over a decade ago by Aharonov and Naveh~\cite{AN02}, we still have an unfortunately small number of complete problems for it. In particular, to the best of our knowledge, the following is an exhaustive list at the time of writing:
\begin{itemize}
	\item Does a given local Hamiltonian have an efficiently preparable ground state~\cite{WJB03}?
	\item Does a given quantum circuit act almost as the identity on computational basis states~\cite{WJB03}?
	\item Given a braid, can it be conjugated by another braid from a given class such that
	the Jones polynomial of its plat closure is nearly maximal~\cite{WY08}?
	\item Given a continuous-time classical random walk on a restricted class of graphs,
	and time $T$, do there exist vertices $i$ and $j$ such that the difference of the
	probabilities of being at $i$ and $j$ is at least $c\cdot\exp(-\mu T)$~\cite{JW06}?
	\item Given a quantum circuit $C$ accepting a non-empty monotone set, what is the
	smallest Hamming weight string accepted by $C$~\cite{GK12}?
\end{itemize}
In this regard, the pursuit of natural complete problems for QCMA has arguably proven rather difficult. Our results add a new, physically-motivated problem to the short list of QCMA-complete problems.

Second, a common focus in quantum complexity theory has been the problem of estimating the ground state energy of a given local Hamiltonian (see, e.g.~\cite{GLSW14} for a survey). However, less attention has been given to the complexity of determining other properties of local Hamiltonians. For example, Brown, Flammia, and Schuch showed~\cite{BFS11} that computing the ground state degeneracy and density of states for a local Hamiltonian is $\class{\#BQP}$-complete. Gharibian and Kempe showed~\cite{GK12} that determining the smallest subset of interaction terms of a given local Hamiltonian which yields a high energy ground space is cq-${\rm\Sigma_2}$-complete. Ambainis has shown~\cite{A14} (among other results) that evaluating local observables against a local Hamiltonian is  ${\rm{P}^{QMA[\log n]}}$-complete, and that determining the spectral gap of a local Hamiltonian is in ${\rm{P}^{QMA[\log n]}}$. Continuing in this vein, our work initiates a new direction of study regarding properties of local Hamiltonians beyond estimating the ground state energy, namely the study of ground state connectivity.

Finally, regarding the use of our proof techniques in the study of quantum algorithms and verification procedures, we hope the Traversal Lemma may prove useful in its own right. For example, in quantum adiabatic algorithms, it is {often} notoriously difficult to understand how a quantum state evolves in time from an easy-to-prepare initial state to some desired final state. The Traversal Lemma gives us a tool for studying the behaviour of such evolutions, playing a crucial role in our analysis here. We remark, however, that in quantum adiabatic evolution, the Hamiltonian itself changes with time, whereas here our Hamiltonian is fixed and we apply local unitary gates to our quantum state.


\paragraph{Organization.} This paper is organized as follows. In Section~\ref{sect:defns}, we state relevant notation, definitions, and useful known results. Section~\ref{sect:nets} constructs $\epsilon$-nets and $\epsilon$-pseudo-nets over unitary operators, {which are} used in Sections~\ref{sect:GSCON_QCMA},~\ref{sect:GSCON_pspace}
and~\ref{sect:GSCON_nexp}
for showing containment of \gscon~in
\class{QCMA},
\class{PSPACE}, and
\class{NEXP}, respectively. Section~\ref{sect:traversal} introduces the notion of $k$-orthogonality and states and proves the Traversal Lemma, which is used in Sections~\ref{sect:GSCON_QCMA},~\ref{sect:GSCON_pspace}, and~\ref{sect:GSCON_nexp} to show \class{QCMA}-hardness, \class{PSPACE}-hardness, and  \class{NEXP}-hardness of \gscon. Section~\ref{sect:tightness} shows our result regarding tightness of the Traversal Lemma and Section~\ref{app:prop} studies the properties of $k$-orthogonality further. We conclude and state open problems in Section~\ref{sect:conclusion}.

\section{Preliminaries} \label{sect:defns}

\paragraph{Notation.} The notation $:=$ is used to indicate a definition. Given $x\in\set{0,1}^n$, $\ket{x}\in(\complex^2)^{\otimes n}$ denotes the computational basis state labeled by $x$. For a vector $\ket{v}$, define its Euclidean norm as $\enorm{\ket{v}}:=(\sum_i \abs{v_i}^2)^{1/2}$ and its infinity norm as $\snorm{\ket{v}}:=\max_{i}\abs{v_i}$. For complex Euclidean space $\spa{X}$, let $\lin{\spa{X}}$, $\herm{\spa{X}}$ and $\unitary{\spa{X}}$ denote the sets of linear, Hermitian and unitary operators acting on $\spa{X}$, respectively. We use the following matrix norms: $\mnorm{A}:=\max_{ij}\abs{A(i,j)}$, the spectral norm $\snorm{A} := \max\{\norm{A\ket{v}}_2 : \norm{\ket{v}}_2 = 1\}$, {the} trace norm $\trnorm{A}:=\trace{\sqrt{A^\dagger A}}$, {and the Frobenius norm
$\norm{A}_{\mathrm{F}} := \sqrt{\trace(A^\dagger A)}$}. The Hilbert-Schmidt or trace inner product between operators $A$ and $B$ is $\langle A,B\rangle := \trace(A^\dagger B)$.  The set of natural numbers is $\natural$, and $[m]:=\set{1,\ldots, m}$. Throughout this paper, we treat the local dimension $d$ of quantum systems as a constant.


\paragraph{Definitions.} We now formally define the problem studied in this paper. (To ease parsing of the definition, the input parameters are highlighted in maroon online.)

\begin{definition}[Ground State Connectivity (\gscon $(\red{H},\red{k},\red{\aaa},\red{\bbb},\red{\ccc},\red{\ddd},\red{\Delta},\red{l},\red{m},\red{U_\psi},\red{U_\phi})$)]\label{def:GSCON}
~
\begin{mylist}{\parindent}
\item Input parameters:
    \begin{enumerate}
        \item \red{$k$}-local Hamiltonian $\red{H}=\sum_i H_i$ acting on $n$ qubits with $H_i \in \herm{(\complex^2)^{\otimes \red{k}}}$ satisfying $\snorm{H_i}\leq 1$.

        \item $\red{\aaa},\red{\bbb},\red{\ccc},\red{\ddd}, \red{\Delta}\in\real$, and integer $\red{m}\geq0$, such that $\red{\bbb}-\red{\aaa}\geq \red{\Delta}$ and $\red{\ddd}-\red{\ccc}\geq \red{\Delta}$.

        \item Polynomial size quantum circuits \red{$U_\psi$} and \red{$U_\phi$} generating ``starting'' and ``target'' states $\ket{\psi}$ and $\ket{\phi}$ (starting from $\ket{0}^{\otimes n}$), respectively, satisfying $\bra{\psi}\red{H}\ket{\psi}\leq \red{\aaa}$ and $\bra{\phi}\red{H}\ket{\phi}\leq \red{\aaa}$.
    \end{enumerate}
\item Output:
\begin{enumerate}
    \item If there exists a sequence of \red{$l$}-local unitaries $(U_{i})_{i=1}^m \in\unitary{\complex^2}^{\times \red{m}}$ such that:
    \begin{enumerate}
        \item (Intermediate states remain in low energy space) For all $i\in [\red{m}]$ and intermediate states ${\ket{\psi_i}:=U_i\cdots U_2U_1\ket{\psi}}$, one has $\bra{\psi_i}\red{H}\ket{\psi_i}\leq \red{\aaa}$, and
        \item (Final state close to target state) $\enorm{ U_{\red{m}} \cdots U_1 \ket{\psi}-\ket{\phi}} \leq \red{\ccc}$,
    \end{enumerate}
    then output YES.
    \item If for all $\red{l}$-local sequences of unitaries $(U_{i})_{i=1}^{\red{m}}\in\unitary{\complex^2}^{\times \red{m}}$, either:
    \begin{enumerate}
        \item (Intermediate state obtains high energy) There exists $i\in [\red{m}]$ and an intermediate state ${\ket{\psi_i}:=U_i\cdots U_2U_1\ket{\psi}}$, such that $\bra{\psi_i}\red{H}\ket{\psi_i}\geq \red{\bbb}$, or
        \item (Final state far from target state) $\enorm{ U_{\red{m}} \cdots U_1 \ket{\psi}-\ket{\phi}} \geq \red{\ddd}$,
    \end{enumerate}
    then output NO.
\end{enumerate}
\end{mylist}
\end{definition}

\noindent A few remarks are in order. First, in the Hamiltonian complexity literature the gap size $\Delta$ for energy levels of local Hamiltonians is often taken to be inverse polynomial. Some of our results require this gap to be exponentially small. Allowing $\Delta$ to be specified as input thus allows us to precisely formulate such results. Second, the circuits $U_\psi$ and $U_\phi$ are assumed to be given in terms of $1$ and $2$-qubit unitary gates. Third, all input parameters are specified with rational entries, each using $O(\poly(n))$ bits of precision. Fourth, as alluded to in the introduction, one can consider the special case of \gscon\ in which all states $\ket{\psi_i}$ are \emph{exactly} in the ground space of $H$; let us briefly define this variant formally, as our proof techniques for QCMA-completeness (Section~\ref{sect:GSCON_QCMA}) also apply in this special case.

\begin{definition}[Frustration-Free \gscon (\ffgscon $(\red{H},\red{k},\red{\bbb},\red{\ccc},\red{\ddd},\red{\Delta},\red{l},\red{m},\red{U_\psi},\red{U_\phi})$)]\label{def:FFGSCON}
Defined as \gscon\ with positive semidefinite $H$ and $\aaa=0$ (i.e. $H$ is frustration-free and the starting state $\ket{\psi}$, final state $\ket{\phi}$, and all intermediate states are \emph{exactly} in the ground space of $H$.)
\end{definition}

For completeness, we next give a formal definition of the complexity class \class{QCMA}~\cite{AN02} (also known as Merlin-Quantum-Arthur (MQA)~\cite{W09_2}.

\begin{definition}[\class{QCMA}]\label{def:QCMA}
    A promise problem $A=(\ayes,\ano)$ is in \class{QCMA} if and only if there exist polynomials $p$, $q$ and a polynomial-time uniform family of quantum circuits $\set{Q_n}$, where $Q_n$ takes as input a string $x\in\Sigma^*$ with $\abs{x}=n$, a classical proof ${y}\in \set{0,1}^{\otimes p(n)}$, and $q(n)$ ancilla qubits in state $\ket{0}^{\otimes q(n)}$, such that:
    \begin{itemize}
    \item (Completeness) If $x\in\ayes$, then there exists a proof $y\in\set{0,1}^{\otimes p(n)}$ such that $Q_n$ accepts $(x,{y})$ with probability at least $2/3$.
    \item (Soundness) If $x\in\ano$, then for all proofs ${y}\in \set{0,1}^{\otimes p(n)}$, $Q_n$ accepts $(x,{y})$ with probability at most $1/3$.
    \end{itemize}
\end{definition}


\paragraph{Useful known results.} We next state known results which prove useful in this paper. The first of these is the Gentle Measurement Lemma of Winter~\cite{W99}; the specific variant we state below is Lemma 9.4.2 from the textbook of Wilde~\cite{W13}.

\begin{lemma}[Gentle Measurement Lemma~\cite{W99}, as stated in Lemma 9.4.2 of~\cite{W13}]\label{l:gentle}
    Let $\rho\in\lin{\complex^d}$ be a density operator and $O\preceq\Lambda\preceq I$ a measurement operator for $\Lambda\in\lin{\complex^d}$, such that $\tr(\Lambda\rho)\geq 1-\epsilon$. Then, $\trnorm{\rho-\sqrt{\Lambda}\rho\sqrt{\Lambda}}\leq 2\sqrt{\epsilon}$.
\end{lemma}

We next recall Kempe and Regev's $3$-local circuit-to-Hamiltonian construction~\cite{KR03}, which maps a given quantum circuit $V=V_L\cdots V_1$ (where each $V_i$ is at most $2$-local) acting on a \emph{proof} register (register $A$) and \emph{ancilla} register (register $B$) to a $3$-local Hamiltonian $H$ acting on $A\otimes B\otimes C$, where $C$ is a \emph{clock} register (represented in unary). The precise details of the construction are not necessary for this work; rather, we require only the following key property of $H$. Define the \emph{history state} for arbitrary proof $\ket{\psi}$ in register $A$ as
\begin{equation}\label{eqn:hist}
 \histstate:=\frac{1}{\sqrt{L+1}}\sum_{i=0}^L V_i\cdots V_1 \ket{\psi}_{A}\otimes\ket{0}_B\otimes\ket{i}_C.
\end{equation}
Then, the question of whether $V$ accepts $\ket{\psi}$ is related to the smallest eigenvalue of $H$ as follows.

\begin{lemma}[Kempe and Regev~\cite{KR03}]\label{l:kemperegev}
    Kempe and Regev's construction maps a quantum circuit $V$ to a $3$-local Hamiltonian $H$ with parameters $\alpha$ and $\beta$ satisfying:
    \begin{itemize}
        \item If there exists a proof $\ket{\psi}$ accepted by $V$ with probability at least $1-\epsilon$, then $\histstate$ achieves
            \[ \tr(H\histstateketbra)\leq \alpha:= \epsilon/(L+1).\]
        \item If $V$ rejects all proofs $\ket{\psi}$ with probability at least $1-\epsilon$, then the smallest eigenvalue of $H$ is at least $\beta\in\Omega\left(\frac{1}{L^3}\right)$.
    \end{itemize}
\end{lemma}

We next discuss the classical reconfiguration problem for Boolean formulae known as (s,t)-Connectivity (denoted $\stcon$, for short).

\begin{definition}[$\stcon$] \label{def:STCON}
    Given a Boolean \textup{$3$-CNF} formula $\phi$ and solutions $x,y\in\set{0,1}^n$ to $\phi$, does there exist a sequence of strings $(x_i)_{i=1}^m$ such that
    \begin{enumerate}
        \item $x_1=x$ and $x_m=y$, and
        \item for all $i\in[m]$, the Hamming distance between $x_i$ and $x_{i+1}$ is at most $1$, and
        \item for all $i\in[m]$, $x_i$ is a solution to $\phi$?
    \end{enumerate}
\end{definition}

\begin{theorem}[\cite{GKMP06}]\label{thm:GKMP06}
    $\stcon$~is \class{PSPACE}-complete.
\end{theorem}

Finally, we state a few useful norm inequalities. For arbitrary complex unit vectors $\ket{v}$ and $\ket{w}$ (see, e.g., Equation 1.33 of Reference~\cite{G13}):
\begin{equation}\label{eqn:enorm}
    \trnorm{\ketbra{v}{v}-\ketbra{w}{w}}=2\sqrt{1-\abs{\brakett{v}{w}}^2}\leq 2\enorm{\ket{v}-\ket{w}}.
\end{equation}
For arbitrary (not necessarily normalized) complex vectors, we have:

\begin{equation}\label{eqn:fnormjj}
    \fnorm{\ketbra{v}{v}-\ketbra{w}{w}} \leq \left( \norm{\ket{v}}_2 + \norm{\ket{w}}_2 \right) \; \norm{\ket{v} - \ket{w}}_2.
\end{equation}

\begin{proof}
	We use the triangle inequality and the fact that $\fnorm{\ketbra{a}{b}}=\enorm{\ket{a}}\enorm{\ket{b}}$ (seen by expanding the definition of $\fnorm{\ketbra{a}{b}}$) to obtain:
	\begin{eqnarray*}
		\fnorm{\ketbra{v}{v}-\ketbra{w}{w}}&\leq& \fnorm{\ketbra{v}{v}-\ketbra{v}{w}}+\fnorm{\ketbra{v}{w}-\ketbra{w}{w}}\\
		&=&
		\fnorm{\ket{v}(\bra{v}-\bra{w})}+                    \fnorm{(\ket{v} - \ket{w})\bra{w}}\\
		&=&\left( \norm{\ket{v}}_2 + \norm{\ket{w}}_2 \right) \; \norm{\ket{v} - \ket{w}}_2.
	\end{eqnarray*}
\end{proof}

\section{Nets and pseudo-nets over unitary operators}\label{sect:nets}

In order to show containment of \gscon~in the complexity classes of interest, we require nets with respect to spectral norm over unitary operators. In this section, we give two types of nets: (1) An $\epsilon$-net over single qubit unitaries (Lemma~\ref{l:net}), and (2) an $\epsilon$-pseudo-net over unitaries of any dimension $d\geq 2$
 (Lemma~\ref{l:pseudonet}). The former is used in Lemma~\ref{l:inPSPACE} (containment in \class{PSPACE}) and Lemma~\ref{l:inNEXP} (containment in \class{NEXP}), and consists strictly of unitary operators. The latter is used in Lemma~\ref{l:inQCMA} (containment in \class{QCMA}), and is a relaxation of a net in that it contains \emph{non-unitary} operators; this relaxed definition, however, allows for a straightforward construction in dimensions greater than two. Note that having an exact net helps make the analysis in the proof of Lemma~\ref{l:inPSPACE} easier, explaining why we use both kinds of nets.
We begin with a simple single-qubit $\epsilon$-net construction.

\begin{lemma}\label{l:net}
    For any $0<\epsilon \leq 1$, there exists an $\epsilon$-net with respect to the spectral norm over $\unitary{\complex^2}$ of size $O(\epsilon^{-8})$. Moreover, given the index $i$ of any element $U_i$ in the net, $U_i$ can be computed in time $O(\log^2(1/\epsilon))$.
\end{lemma}

The proof is given in Appendix~\ref{app:netproofs}, and relies on a simple characterization of single qubit unitaries. For larger dimensions $d>2$, however, we are unaware of a similar characterization. Thus, for $d>2$ we construct\footnote{It was pointed out to us by an anonymous referee that there is an alternative way to construct an $\epsilon$-net over unitary operators with $d>2$, which can be used in place of our pseudo-net here. Namely, one casts a net over the set of Hermitian operators $H$ satisfying $\snorm{H}\leq \pi$, and subsequently exponentiates the items in the net.} an \emph{$\epsilon$-pseudo-net}.} Intuitively, a pseudo-net over unitary operators contains matrices which are close to, but not necessarily, unitary. However, to aid in its use, it has two important properties: First, we give an efficient ``check'' procedure $C$ such that, for any unitary $U$, there exists a net element $M$ satisfying $\snorm{U-M}\leq \epsilon$ and such that $M$ is accepted by $C$. Second, we give an efficient ``rounding'' procedure $R$ such that if net element $M$ is accepted by $C$, then $R$ rounds $M$ to a unitary $U$ satisfying $\snorm{U-M}\leq \epsilon$.

\begin{definition}[$\epsilon$-pseudo-net]\label{def:pseudonet}
    Let $S\subseteq \lin{\complex^d}$. Then, we call $S'\subseteq \lin{\complex^d}$ an $\epsilon$-pseudo-net over $S$ if there exist $O(\poly(d))$-time algorithms $C$ (for \emph{checking}) and $R$ (for \emph{rounding}) taking as input $M\in\lin{\complex^d}$ such that:
    \begin{enumerate}
        \item (Checking) $\forall\; M\in S$, there exists $M'\in S'$ such that $C$ accepts $M'$ and $\snorm{M-M'}\leq \epsilon$.
        \item (Rounding) $\forall\; M'\in S'$, if $C$ accepts $M'$, then algorithm $R$ maps $M'$ to $M\in S$ such that $\snorm{M-M'}\leq \epsilon$.
    \end{enumerate}
\end{definition}
We {show in Appendix~\ref{app:netproofs}} that there is a straightforward way to construct an $\epsilon$-pseudo-net over $S=\unitary{\complex^d}$ for any $d\geq 2$. The {ideas} are based on a standard construction for nets over unitary operators, as used in Reference~\cite{PGACHW11} and detailed further in Lemma 7.13 of Reference~\cite{G13}; this standard construction is, however, inherently non-explicit. Thus, we {adapt it}  as necessary to obtain an \emph{explicit} $\epsilon$-pseudo-net.

\begin{lemma}\label{l:pseudonet}
    For any $0<\epsilon<1$, there exists a set $N\subseteq \lin{\complex^d}$ of size $O(d^{7}/\epsilon^2)$ such that:
    \begin{enumerate}
        \item $N$ is an $\epsilon$-pseudo-net with respect to spectral norm over unitaries $\unitary{\complex^d}$.
        \item Given index $i\in\set{1,\ldots,\abs{N}}$, the $i$'th operator $\widetilde{U}_i$ in the net can be computed in time $O(d^2\log^2(d^{5/2}/\epsilon))$. Here, by $i$'th operator, we mean with respect to a fixed canonical ordering set by the construction of $N$.
    \end{enumerate}
\end{lemma}

\section{$k$-Orthogonality and the Traversal Lemma} \label{sect:traversal}

{The key technical tool for proving our hardness results is the Traversal Lemma (Lemma~\ref{l:traversal}), which we state and prove in this section. In Sections~\ref{sect:tightness} and~\ref{app:prop}, we then show that this lemma is tight up to a polynomial factor and give a further study into the notion of $k$-orthogonality, respectively. We begin by introducing the notions of \emph{$k$-orthogonal states} and \emph{$k$-orthogonal subspaces}.}

\begin{definition}[$k$-orthogonal states and subspaces]\label{def:korth}
For $k \geq 1$, a pair of states ${\ket{v},\ket{w}\in(\complex^d)^{\otimes n}}$ is \emph{$k$-orthogonal} if for all $k$-qudit unitaries $U$, we have $\bra{w}U\ket{v}=0$. We call subspaces $S,T \subseteq(\complex^d)^{\otimes n}$ \emph{$k$-orthogonal} if any pair of vectors $\ket{v}\in S$ and $\ket{w}\in T$ are $k$-orthogonal.
\end{definition}

\noindent Let us comment on the structure of $k$-orthogonal states. First, $k$-orthogonality implies orthogonality, but not vice versa. For example, $\ket{000}$ and $\ket{111}$ are $2$-orthogonal and hence orthogonal. In contrast, $\ket{000}$ and $\ket{100}$ are orthogonal but not $k$-orthogonal for any $k\geq 1$ (i.e. simply apply Pauli $X$ to qubit $1$ to map $\ket{000}$ to $\ket{100}$). Similarly, letting $S$ and $T$ denote the $+1$ eigenspaces of $I \otimes \kb{000}$ and $I \otimes \kb{111}$, respectively, we have that $S$ and $T$ are $2$-orthogonal subspaces.

We now prove the Traversal Lemma, which says the following: For any two $k$-orthogonal subspaces $S$ and $T$ with $\ket{v} \in S$ and $\ket{w} \in T$, any sequence of $m$ $k$-qudit unitaries mapping $\ket{v}$ to $\ket{w}$ must induce an evolution which has ``large'' overlap with the orthogonal complement of both $S$ and $T$ at some time step $i\in[m]$.

\begin{lemma}[Traversal Lemma]\label{l:traversal}
    Let $S,T \subseteq(\complex^d)^{\otimes n}$ be $k$-orthogonal subspaces. Fix arbitrary states $\ket{v}\in S$ and $\ket{w}\in T$, and consider a sequence of $k$-qudit unitaries $(U_i)_{i=1}^m$ such that
    \[
        \norm{\ket{w}- U_m\cdots U_1 \ket{v}}_2 \leq \epsilon
    \]
for some $0\leq \epsilon < 1/2$. Define $\ket{v_i}:=U_i \cdots U_1 \ket{v}$ and $P:=I-\Pi_S-\Pi_T$. Then, there exists an $i\in[m]$ such that
    \[
        \bra{v_i} P \ket{v_i}\geq \left(\frac{1- 2\epsilon}{2m}\right)^2.
    \]
\end{lemma}
\begin{proof}

    We give a proof by contradiction. Suppose
    that for all $i\in[m]$, {the inner products satisfy} $\bra{v_i}P\ket{v_i}<\delta:=[(1-  2\epsilon)/(2m)]^2$.
     Consider the following thought experiment inspired by the quantum Zeno effect. Imagine that after each $U_i$ is applied, we measure $\ket{v_i}$ using the projective measurement $(\Pi, I-\Pi)$ for $\Pi:=I-P$, and postselect on obtaining outcome $\Pi$. {Define the following two sequences:}
	
	\begin{itemize}
	\item $\ket{v_i'}:=\Pi\ket{v_i}$ for $i \in [m]$,
	\item $\ket{v_1''} := \ket{v_1'}$ and $\ket{v_i''}:=\Pi \, U_i \ket{v_{i-1}''}$ for $i \in \set{2, \ldots, m}$.
	\end{itemize}
	Note that $\ket{v_i'}$ and $\ket{v_i''}$ are not necessarily normalized.

	To set up our contradiction, we first prove by induction on $i$ that
    \begin{equation}\label{eqn:ind}
        \trnorm{\ketbra{v_i}{v_i}-\ketbra{v_i''}{v_i''}}< 2i\sqrt{\delta}.
    \end{equation}
    For the base case $i=1$, we have $\ket{v_1''}=\ket{v_1'}$. Then, since $\bra{v_1}P\ket{v_1}<\delta$, we know that ${\tr(\Pi\ketbra{v_1}{v_1})>1-\delta}$, and so the Gentle Measurement Lemma~\cite{W99} (Lemma~\ref{l:gentle}) yields
    \begin{equation}\label{eqn:gentle}
        \trnorm{\ketbra{v_1}{v_1}-\ketbra{v_1''}{v_1''}}={\trnorm{\ketbra{v_1}{v_1}-\ketbra{v_1'}{v_1'}}<2\sqrt{\delta}},
    \end{equation}
    as required. For the inductive case, assume Equation~(\ref{eqn:ind}) holds for $1\leq i\leq j-1$. We prove it holds for $i=j$. Specifically,
    \begin{eqnarray}
        \trnorm{\ketbra{v_j}{v_j}-\ketbra{v_j''}{v_j''}}&\leq&
        \trnorm{\ketbra{v_j}{v_j}-\ketbra{v'_j}{v'_j}}+\trnorm{\ketbra{v'_j}{v'_j}-\ketbra{v_j''}{v_j''}}\nonumber\\
        &<&
        2\sqrt{\delta}+\trnorm{\ketbra{v'_j}{v'_j}-\ketbra{v_j''}{v_j''}}\nonumber\\
        &=&2\sqrt{\delta}+\trnorm{\Pi U_{j}\left(\ketbra{v_{j-1}}{v_{j-1}}-\ketbra{v_{j-1}''}{v_{j-1}''}\right)U_{j}^\dagger\Pi}\nonumber\\ &\leq&2\sqrt{\delta}+\trnorm{\ketbra{v_{j-1}}{v_{j-1}}-\ketbra{v_{j-1}''}{v_{j-1}''}}\nonumber\\        &<&2\sqrt{\delta}+2(j-1)\sqrt{\delta}\nonumber\\
        &=&2j\sqrt{\delta},\label{eqn:small1}
    \end{eqnarray}
    where the first statement follows from the triangle inequality, the second from the Gentle Measurement Lemma, the fourth from the facts that the Schatten $p$-norms are invariant under isometries and that ${\norm{ABC}_p\leq\snorm{A}\norm{B}_p\snorm{C}}$~\cite{W08_2}, and the fifth from the induction hypothesis. This establishes Equality~(\ref{eqn:ind}).
	
	We thus have
	\begin{eqnarray}
        \trnorm{\ketbra{v_m''}{v_m''}-\ketbra{w}{w}} &\leq&   \trnorm{\ketbra{v_m''}{v_m''}-\ketbra{v_m}{v_m}}+\trnorm{\ketbra{v_m}{v_m}-\ketbra{w}{w}}\nonumber\\
        &<& 2m \sqrt{\delta} + 2\epsilon\nonumber\\
        &=& 1, \label{eqn:contrabound1}
    \end{eqnarray}
    where we have used Equation~(\ref{eqn:enorm}) to bound
	\[
		\trnorm{\ketbra{v_m}{v_m}-\ketbra{w}{w}} \leq 2 \norm{ \ket{v_m} - \ket{w}}_2 \leq 2 \epsilon.
	\]
We are now ready to obtain the desired contradiction.

To do so, observe that since $\ket{v} \in S$, and since $S$ and $T$ are $k$-orthogonal subspaces, we have that for all $i \in [m]$, $\ket{v_i''}\in S$ (i.e., if $S$ is $1$-dimensional, this is the Zeno effect). Thus, we have $\brakett{v_m''}{w}=0$, implying that
\[
\trnorm{ \kb{v_m''} - \kb{w} } = 1 + \norm{ \ket{v_m''} }_2 \geq 1.
\]
This contradicts Equation~(\ref{eqn:contrabound1}), as desired.
\end{proof}

\section{\class{QCMA}-completeness} \label{sect:GSCON_QCMA}

In this section, we prove the following theorem.

\begin{theorem}\label{thm:qcmacomplete}
    There exists a polynomial $p$ such that \gscon~is \class{QCMA}-complete for ${m\in O(p(n))}$, $\Delta\in{\Theta} (1/m^5)$, $l=2$, and $k\geq 5$, where $n$ denotes the number of qubits $H$ acts on.
\end{theorem}

\noindent{Remarks: Intuitively, this says that \gscon~is \class{QCMA}-complete when the unitaries $U_i$ are at most $2$-local, the number of unitaries scales polynomially, and the gap $\Delta$ scales inverse polynomially. Note that our proof in fact shows a stronger result than stated above: Recalling that \ffgscon\ (Definition~\ref{def:FFGSCON}) is the special case of \gscon\ in which $H$ is frustration-free and the starting state $\ket{\psi}$, final state $\ket{\phi}$, and all intermediate states are \emph{exactly} in the ground space of $H$ (as opposed to being low-energy states in the style of the original definition of the local Hamiltonian problem~\cite{KSV02}), our proof shows that \ffgscon\ (with the same parameter range as in Theorem~\ref{thm:qcmacomplete} except now $k\geq 7$) is also QCMA-complete. This is because, without loss of generality, one may assume in our QCMA-hardness reduction that the QCMA verifier we start with has perfect completeness\footnote{{The perfect-completeness QCMA construction of~\cite{JKNN12} assumes the verifier uses gates from a specific universal gate set including the Hadamard gate, which has irrational entries. Thus, in the definition of \ffgscon, we would instead allow the input to be specified using (e.g.) a quadratic field extension $\mathbb{F}:\mathbb{Q}$~\cite{cohen93}, as opposed to just rational entries as for \gscon.}}~\cite{JKNN12} (further details given in the proof of Lemma~\ref{l:qcmahard}).}

To prove Theorem~\ref{thm:qcmacomplete}, we prove \class{QCMA}-hardness and containment in \class{QCMA} separately. We begin with \class{QCMA}-hardness.

\subsection{\class{QCMA}-hardness}

We now show that \gscon~is \class{QCMA}-hard in the {regime described below.}

\begin{lemma}\label{l:qcmahard}
    There exists a polynomial $p$ such that \gscon~is \class{QCMA}-hard for $m \in O(p(n))$, $\Delta \in {O}(1/m^5)$, $l=2$, and $k \geq 5$, where $n$ denotes the number of qubits $H$ acts on.
\end{lemma}
\begin{proof}
    At a high level, our approach is as follows. Given a \class{QCMA} verification circuit, let $H^{\rm KR}$ be the $3$-local Hamiltonian output by Kempe and Regev's circuit-to-Hamiltonian construction. Then, our aim is to construct another Hamiltonian $H$ such that ``traversing the ground space of $H$'' forces one to simulate the following protocol --- starting with an initial state of all zeroes:
    \begin{enumerate}
         \item Apply a sequence of $2$-qubit gates to prepare a ground state $\ket{\psi_{H^{\rm KR}}}$ of $H^{\rm KR}$.
         \item Flip a first ``\emph{GO}'' qubit to initiate a ``check'' that $\ket{\psi_{H^{\rm KR}}}$ is indeed a ground state of $H^{\rm KR}$.
         \item Flip a second {and third} ``\emph{GO}'' qubit to end the ``check''.
         \item Uncompute $\ket{\psi_{H^{\rm KR}}}$ to obtain a target state which is all zeroes, except for the ``\emph{GO}'' qubits, which are set to all ones.
    \end{enumerate}

    Formally, let $\Pi'$ be an instance of a \class{QCMA} problem with verification circuit $V'$ acting on a classical proof register $p$ and ancilla register $a$ consisting of $n_p$ and $n_a$ qubits, respectively. Using standard error reduction via parallel repetition, we may assume without loss of generality that $V'$ accepts (rejects) in the YES (NO) case with probability at least $\pacc\geq 1-2^{\enc{\Pi'}}$ ($\prej\geq 1-2^{\enc{\Pi'}}$), where ${\enc{\Pi'}}$ denotes the encoding length of $\Pi'$.

    Let $V$ denote a new circuit which first measures the proof register in the computational basis, and then runs $V'$. (A similar trick is used in~\cite{WJB03}; it directly ensures that the Hamiltonian $H$ we construct shortly has no low energy states of low complexity in the NO case by forcing \emph{all} eigenvalues of $H$ to be large in the NO case.) Formally, $V$ has the following properties: (1) $V$ has $n_a+n_p$ ancilla qubits initialized to all zeroes, (2) in time step $i\in[n_p]$, $V$ applies a CNOT gate with the $i$'th proof qubit as control and ancilla qubit $n_a+i$ as target, and (3) starting at time step $n_p+1$, $V$ simulates $V'$ while acting on register $p$ and the first $n_a$ qubits of $a$. A straightforward argument shows that $V$ accepts a proof if and only if $V'$ does. Moreover, unlike $V'$, the principle of deferred measurement~\cite{NC00} yields that $V$ is sound against a cheating prover which does not send a classical string $x$ as a proof.

    Next, we define our Hamiltonian $H$ based on $V$. Let $H^{\rm KR}$ denote the $3$-local Hamiltonian obtained from $V$ using Kempe and Regev's circuit-to-Hamiltonian construction~\cite{KR03}. Then, we define $H$ to act on a \emph{Hamiltonian} register denoted $h$ and \emph{GO} register denoted $G$. Specifically,
    \[
        H\in\herm{(\complex^2)^{\otimes{(2n_p+n_a+n_c)}}\otimes(\complex^2)^{\otimes 3}},
    \]
    where $n_c$ denotes the polynomial number of qubits used for the clock register of $H^{\rm KR}$,  and
    \begin{equation}\label{eqn:H}
        H:=H^{\rm KR}_h\otimes P_G \quad\quad\text{for}\quad\quad P:=I-\ketbra{000}{000}-\ketbra{111}{111}.
    \end{equation}
    Noting that $P$ can be written $2$-locally as
    \begin{eqnarray*}
        P &=& \frac{1}{2} ( \ketbra{01}{01}\otimes I+  \ketbra{10}{10}\otimes I + I\otimes \ketbra{01}{01}+I\otimes\ketbra{10}{10}
+ \\&&\mbox{\hspace{4mm}}\kb{1} \otimes I \otimes \kb{0} + \kb{0} \otimes I \otimes \kb{1}
        ),
    \end{eqnarray*}
    we have that $H$ is $5$-local. We define our initial and final states as
    \begin{equation}\label{eqn:states}
        \ket{\psi}:=\ket{0}^{\otimes (2n_p+n_a+n_c)} \ket{0}^{\otimes 3}\quad\quad\text{and}\quad\quad\ket{\phi}:=\ket{0}^{\otimes (2n_p+n_a+n_c)} \ket{1}^{\otimes 3}.
    \end{equation}
    Finally, letting $W$ denote a unitary circuit of size $\abs{W}$ which prepares the history state of $H$ given classical proof $x$, define $m:=2(n_p+\abs{W}+1)$. Note that $m$ is polynomial in the input size, since for any YES instance $\Pi$, $V'$ accepts a \emph{classical} proof, and hence the history state for $H^{\rm KR}$ can be prepared in polynomial time. (This observation was also made in~\cite{WJB03}.)
  Set $\ccc=0$, ${\ddd=1/4}$, $\aaa=\alpha$, and ${\bbb=\beta/(16m^2)}$, where $\alpha$ and $\beta$ come from Lemma~\ref{l:kemperegev}. Thus, $\aaa\in O(2^{-{\enc{\Pi'}}})$ and $\bbb\in \Omega(1/m^5)\in \Omega (1/\poly(\enc{\Pi'}))$ (where we have used the facts that $m \geq L$ for $L$ the number of gates in circuit $V$ and $m\in\poly({\enc{\Pi'}})$). Choose $\Delta\in {O}(1/m^5)$ and set $l=2$. Observe that $\Pi=(H,\aaa,\bbb,\ccc,\ddd,\Delta,l,m,\ket{\psi},\ket{\phi})$ is a valid instance of \gscon~which can be computed in polynomial time given $\Pi'=(V')$, as desired.

    We now show correctness. Suppose there exists a proof $x\in\set{0,1}^{n_p}$ accepted by $V$. We demonstrate a sequence $(U_i)_{i=1}^m$ of $2$-qubit unitaries mapping $\ket{\psi}$ to $\ket{\phi}$ through the ground space of $H$. First, note that $\ket{\psi}$ and $\ket{\phi}$ are in the null space of $H$, and hence $\bra{\psi}H\ket{\psi}\leq \aaa$ and $\bra{\phi}H\ket{\phi}\leq \aaa$, as required. Next, recall in Kempe and Regev's construction that the Hamiltonian register $h$ is itself composed of three sub-registers $h_1$, $h_2$, and $h_3$, corresponding to the \emph{proof}, \emph{ancilla}, and \emph{clock} registers for $H$, respectively. The desired sequence $(U_i)_{i=1}^m$ is then given as follows:
    \begin{enumerate}
        \item Apply Pauli $X$ gates to $h_1$ to prepare classical proof $x$, i.e., map $\ket{0}^{\otimes n_p}$ to $\ket{x}$.
        \item Apply $W$ to $h$ to prepare the history state $\ket{\text{hist}_x}$ of $H^{\rm KR}$.
        \item Apply $(X\otimes X\otimes I)_G$ to ``initiate'' checking of $\ket{\text{hist}_x}$.
        \item Apply $(I\otimes I\otimes X)_{G}$ to ``complete'' checking of $\ket{\text{hist}_x}$.
        \item Apply $W^\dagger$ to $h$ to uncompute $\ket{\text{hist}_x}$.
        \item Apply $X$ gates to $h_1$ to map the initial proof $\ket{x}$ back to $\ket{0}^{\otimes n_p}$.
    \end{enumerate}
    Note first that the length of the sequence above is at most $2(n_p+\abs{W}+1)$ gates, as desired. Second, the final state is equal to $\ket{\phi}$, and every intermediate state is in the null space of $H$ except for possibly after Step 3. As for after Step 3, let $\ket{a_3}$ denote our state at this point. Then, since a valid history state $\ket{\text{hist}_x}$ obtains energy $\bra{\text{hist}_x} H^{\rm KR} \ket{\text{hist}_x}\leq \alpha$, we have $\bra{a_3}H\ket{a_3}\leq\alpha=\aaa$, as desired. Thus, if $\Pi'$ is a YES instance, then $\Pi$ is a YES instance of \gscon. For clarity, note that the $G$ register consists of $3$ qubits (instead of $2$), since otherwise a two-qubit unitary can map $\ket{\psi}$ to $\ket{\phi}$ in a single step, bypassing the initiation of the checking of $\ket{\text{hist}_x}$ as in Step 3 above.

{
    Conversely, suppose $\Pi'$ is a NO instance, i.e., for all $x\in\set{0,1}^{n_p}$, $V$ rejects with high probability. Then, by Lemma~\ref{l:kemperegev}, the smallest eigenvalue of $H^{\rm KR}$ is at least $\beta$. Now, let $S$ and $T$ denote the $+1$ eigenspaces of projections $I_h \otimes \kb{000}_G$ and $I_h \otimes \kb{111}_G$, respectively. Observe that $S$ and $T$ are $2$-orthogonal subspaces, and that $\ket{\psi} \in S$ and $\ket{\phi} \in T$. Thus, for any sequence of two-qubit unitaries $(U_i)_{i=1}^m$, either $\norm{\ket{\psi_m} - \ket{\phi}}_2 \geq  1/4=\ddd $ (in which case we have a NO instance of \gscon~and we are done), or we can apply the Traversal Lemma (Lemma~\ref{l:traversal}) with $\epsilon=1/4$ to conclude that there exists an $i\in[m]$ such that
    \[
    	\bra{\psi_i} P' \ket{\psi_i} \geq \left( \dfrac{1}{4m} \right)^2 = \dfrac{\eta_2}{\beta},
    \]
where we define ${\ket{\psi_i}:=U_i\cdots U_1\ket{\psi}}$ and $P' = I - \Pi_S - \Pi_T$. Note that we can write $P'$ as $I_h \otimes P$.
We conclude that
\[
        \bra{\psi_i} H \ket{\psi_i}
=         \bra{\psi_i} H^{\rm KR} \otimes P \ket{\psi_i}
\geq     \beta     \bra{\psi_i} I_h \otimes P \ket{\psi_i}
=     \beta     \bra{\psi_i} P' \ket{\psi_i}
 \geq \eta_2,
    \]
where the first inequality follows since $H^{\rm KR}\succeq \beta I$.
}

{Finally, as alluded to in the remarks below Theorem~\ref{thm:qcmacomplete}, without loss of generality, the original QCMA verifier $V$ we started with can be assumed to have perfect completeness~\cite{JKNN12}. In this case, if we instead use\footnote{{The $3$-local Kempe and Regev~\cite{KR03} construction has non-positive terms in its propagation Hamiltonian which are not minimized by history states; thus, unlike Kitaev's $5$-local construction, it does not give rise to a frustration-free $H$ for a YES instance.}} Kitaev's original $5$-local circuit-to-Hamiltonian construction to define $H^{\rm KR}$, then for a YES instance $H$ here is $7$-local and frustration-free, and $\ket{\psi}$, $\ket{\phi}$, and all intermediate states $\ket{\psi_i}$ lie exactly in the ground space of $H$. Thus, we obtain QCMA-hardness of \ffgscon, i.e. \gscon\ is QCMA-hard even if in the YES case, we require that all $\ket{\psi_i}$ lie \emph{exactly} in the ground space of a frustration-free Hamiltonian.}
\end{proof}

\paragraph{Remark.} There is no loss of generality in restricting ourselves to $2$-qubit unitaries in the proof above. Specifically, the same proof applies almost identically if we instead allow $p$-qubit unitaries for any constant $p\geq2$ by changing Equation~(\ref{eqn:states}) to
\[
    \ket{\psi}:=\ket{0}^{\otimes (2n_p+n_a+n_c)} \ket{0}^{\otimes (p+1)}\quad\quad\text{and}\quad\quad\ket{\phi}:=\ket{0}^{\otimes (2n_p+n_a+n_c)} \ket{1}^{\otimes (p+1)},
\]
i.e., the \emph{GO} register consists more generally of $p+1$ qubits. Note that the \travlemmatwo~still applies in this more general setting, and second, the projector $P$ onto the \emph{GO} register can be represented as a $2$-local Hamiltonian regardless of the value of $p$, implying we still have ${k=5}$.

\paragraph{Remark.} In the proof of Theorem~\ref{l:qcmahard}, we used Kempe and Regev's $3$-local circuit-to-Hamiltonian construction. One might ask whether one of the known $2$-local constructions based on perturbation theory gadgets may instead be applied to reduce the locality of $H$ further. The main issue in doing so is that here we require the ability to construct the ground state efficiently. In other words, the perturbation theory reduction should ideally produce a ground state whose structure is similar to the history state. Now, Oliveira and Terhal~\cite{OT05} have in fact proven such a perturbation theory result in which the resulting $2$-local Hamiltonian's ground space approximates the starting Hamiltonian's ground space. However, we require a stronger statement than this. To explain, let $H$ denote a $k$-local Hamiltonian and $H'$ the $2$-local Hamiltonian resulting from the construction in~\cite{OT05}. Then, our proof requires a statement of the form\footnote{Note that in~\cite{OT05}, $H$ and $H'$ live in different spaces, so our statement here should not be read literally. Rather, it is intended to give a flavor of the ideal behavior we would like the perturbation theory reduction to obey, without getting into finer details in our discussion here.}: If $\bra{v}H\ket{v}\leq a$, then $\bra{v}H'\ket{v}\leq a$, and if $\bra{v}H\ket{v}\geq b$, then $\bra{v}H'\ket{v}\geq b$. Unfortunately, as far as we are aware, it seems the first of these conditions can be violated for the gadgets presented in~\cite{OT05}. Intuitively, what is happening here is that {although $\bra{v}H\ket{v}\leq a$ (i.e. the expectation is ``small''), it may be that $\ket{v}$ does not fully lie in the ground space of $H$}, but rather has some small overlap with a higher energy subspace $S$. If this higher energy space $S$ is then penalized strongly in $H'$, {then} $\bra{v}H'\ket{v}$ can be large.

\subsection{Containment in \class{QCMA}}

We now show that \gscon~with $2$-local unitaries $U_i$ is in \class{QCMA} so long as the gap $\Delta$ scales inverse polynomially and the number of unitaries $m$ scales polynomially with the input size.

\begin{lemma}\label{l:inQCMA}
    For any nonnegative constants $c_1$ and $c_2$, \gscon~is in \class{QCMA} for $\Delta\geq 1/n^{c_1}$, $m\leq n^{c_2}$, $l=2$, and $k\in O(\log n)$, where $n$ denotes the number of qubits $H$ acts on.
\end{lemma}

\begin{proof}
    Let $\Pi=(H,\aaa,\bbb,\ccc,\ddd,\Delta,l,m,\ket{\psi},\ket{\phi})$ be an instance of $\gscon$. The proof system is given below. Steps 4 and 5 follow standard ideas; thus, we simply sketch them here. Let $L$ denote the number of local terms in $H$.
    \begin{algorithm}
    \begin{enumerate}
        \item The prover sends a sequence $(\widetilde{U}_i)_{i=1}^m \subseteq\lin{\complex^2\otimes \complex^2}$ of matrices from the $\epsilon$-pseudo-net of Lemma~\ref{l:pseudonet}, for $\epsilon :=\Delta/16mL$.
        \item (Unitary check) The verifier runs algorithm $C$ from Lemma~\ref{l:pseudonet} on each $\widetilde{U}_i$, and rejects if $C$  rejects.
        \item (Rounding step) The verifier uses algorithm $R$ from Lemma~\ref{l:pseudonet} to construct a sequence ${(V_i)_{i=1}^m \subseteq\unitary{\complex^2\otimes \complex^2}}$ such that for all $i\in [m]$, $\snorm{\widetilde{U}_i-V_i}\leq \epsilon$.
        \item (Low energy check) Define ${\ket{\psi_t}:=V_t\cdots V_1\ket{\psi}}$. For all $t\in[m]$, the verifier prepares state $\ket{\psi_t}$ a polynomial number of times, and runs Kitaev's phase estimation procedure~\textup{\cite{KSV02}} to estimate $\bra{\psi_t}H\ket{\psi_t}$ within inverse polynomial accuracy. The verifier rejects if this estimate is larger than $\eta_1+\epsilon$.
        \item (Close to target state check) The verifier performs the SWAP test~\textup{\cite{BCWW01}} between $\ket{\psi_m}$ and $\ket{\phi}$ polynomially many times to estimate $\enorm{\ket{\psi_m}-\ket{\phi}}$ within inverse polynomial accuracy. The verifier rejects if this estimate is larger than $\eta_3+\epsilon$.
        \item The verifier accepts.
    \end{enumerate}
    \caption{\class{QCMA} proof system for \gscon}
    \end{algorithm}
    \noindent The verifier's action is clearly implementable by a polynomial size quantum circuit.

    We now show correctness. Let $N$ denote the $\epsilon$-pseudo-net over $2$-qubit unitaries from Lemma~\ref{l:pseudonet} (i.e., $d=4$ in Lemma~\ref{l:pseudonet}), for $\epsilon$ as chosen above. Suppose now that $\Pi$ is a YES instance, i.e., there exists a sequence of $2$-qubit unitaries $(U_i)_{i=1}^m$ mapping $\ket{\psi}$ to $\ket{\phi}$ through the ground space of $H$. Then, in Step 1, the prover sends sequence $(\widetilde{U}_i)_{i=1}^m\in N^{\times m}$ such that $\snorm{U_i-\widetilde{U}_i}\leq \epsilon$ for $i\in[m]$. By   Definition~\ref{def:pseudonet} and Lemma~\ref{l:pseudonet}, Step 2 will pass and the conditions of Step 3 will be met with certainty.

Next, we claim that for all $t\in[m]$, $\snorm{{U}_t\cdots{U}_1-{V}_t\cdots{V}_1}\leq 2\epsilon t$. To see this, we first bound
\[
        \snorm{{U}_t\cdots{U}_1-{V}_t\cdots{V}_1}\leq \snorm{{U}_t\cdots{U}_1-\widetilde{U}_t\cdots \widetilde{U}_1}+\snorm{\widetilde{U}_t\cdots\widetilde{U}_1-V_t\cdots V_1}
\]
and use the fact~\cite{NC00} that for any two quantum circuits $U=U_{t}\cdots U_1$ and $V=V_{t}\cdots V_1$ satisfying $\snorm{U_j-V_j}\leq \epsilon$, we have $\snorm{U-V}\leq \sum_{i=1}^t\snorm{U_i-V_i}$. Defining $\ket{u_t}:=U_t\cdots U_1\ket{\psi}$ and recalling that $\ket{\psi_t} := V_t\cdots V_1\ket{\psi}$, it follows that for all $t\in[m]$, $\enorm{\ket{u_t}-\ket{\psi_t}}\leq 2\epsilon m$. Thus,
\[
        \abs{\tr(H\ketbra{u_t}{u_t})-\trace(H\ketbra{\psi_t}{\psi_t})}\leq\snorm{H}\trnorm{\ketbra{u_t}{u_t}-\ketbra{\psi_t}{\psi_t}}\leq 4\epsilon m L,
\]
where recall $L$ denotes the number of local terms in $H$, the first inequality follows from H\"{o}lder's inequality, and the second by Equation~(\ref{eqn:enorm}).

Since we chose $\epsilon=\Delta/16mL$, we have $(\eta_2-4\epsilon m L)-(\eta_1+4\epsilon m L)\geq\Delta/2$ and we also have ${(\eta_4-2\epsilon m)-(\eta_3+2\epsilon m)\geq\Delta/2}$, i.e., the error incurred by using our net $N$ shifts the thresholds which Steps $4$ and $5$ must distinguish between by at most $\Delta/4$ each, leaving gaps of size $\Delta/2$. But $\Delta/2$ is inverse polynomially large; thus, with high probability (i.e., inverse exponentially close to $1$), Steps $4$ and $5$ do not reject. We conclude that with high probability, the verifier accepts, as desired.

Conversely, suppose we have a NO instance. Then, either the verifier rejects in Step 2, or it runs Step 3 to ``round'' the prover's provided matrices into a sequence of unitaries $(V_i)_{i=1}^m$. But by the NO conditions of \gscon, we know that for our choice of $\epsilon$, either Step 4 or Step 5 must now reject with high probability (i.e., inverse exponentially close to $1$).
\end{proof}

\section{\class{PSPACE}-completeness} \label{sect:GSCON_pspace}

In this section, we show the following theorem.

\begin{theorem}\label{thm:pspacecomplete}
    \gscon~is \class{PSPACE}-complete for $m=2^n$, $\Delta=2^{-(2n+4)}$, $l=1$, $k=3$, where $n$ denotes the number of qubits $H$ acts on.
\end{theorem}

\noindent Intuitively, this says that \gscon~is \class{PSPACE}-complete when the unitaries  are  $1$-local, the number of unitaries scales exponentially, and the gap $\Delta$  scales inverse exponentially. To show this, we prove \class{PSPACE}-hardness and containment in \class{PSPACE} separately. We begin with \class{PSPACE}-hardness.

\subsection{\class{PSPACE}-hardness}

We now show \class{PSPACE}-hardness of \gscon~for the case of exponentially many $1$-local unitaries and exponentially small gap $\Delta$.

\begin{lemma}\label{l:PSPACEhard}
    \gscon~is \class{PSPACE}-hard for $k=3$, $\aaa=\ccc=0$, $\bbb=2^{-(2n+4)}$, $\ddd=1/4$, $\Delta=2^{-(2n+4)}$, $l=1$, and $m=2^{n}$, where $n$ denotes the number of qubits $H$ acts on.
\end{lemma}

\begin{proof}
We show a polynomial-time many-one or Karp reduction from~$\stcon$~(which by Theorem~\ref{thm:GKMP06} is \class{PSPACE}-complete) to $\gscon$. Specifically, let $\Pi=(\phi,x,y)$ be an instance of $\stcon$~for \textup{$3$-CNF} $\phi$. The main idea is to embed $\phi$ trivially into a $3$-local Hamiltonian $H$ as follows. For each clause $c_i$ of $\phi$, we define a local Hamiltonian constraint $H_i$ to penalize the unique $3$-bit ``bad'' assignment to $c_i$, i.e., $H_i := \ketbra{z_i}{z_i}$ for $c_i(z_i)=0$. {Setting our parameters as in the theorem statement, we
thus obtain an instance} $\Pi'=(H := \sum_i H_i,\aaa,\bbb,\ccc,\ddd, \Delta, l, m,U_x,U_y)$ of $\gscon$, where $U_x\ket{0\cdots 0}=\ket{x}$ and $U_y\ket{0\cdots 0}=\ket{y}$ for
the strings $x$ and $y$, respectively, from the $\stcon$~instance. Now, given strings $x,y\in\set{0,1}$, it is trivial that if $\Pi$ is a YES instance of $\stcon$, then $\Pi'$ is a YES instance of $\gscon$: Namely, simulate local bit flips on strings by Pauli $X$ gates to map $\ket{x}$ to $\ket{y}$ while staying in the null space of $H$. Note that since there are at most $2^n$ distinct strings on $n$ bits, at most $m=2^n$ Pauli $X$ gates suffice to map $\ket{x}$ to $\ket{y}$.

Conversely, suppose $\Pi$ is a NO instance of $\stcon$.
Let $S$ denote the subspace corresponding to the span of all states $\ket{z}$ such that $z$ can be obtained via a sequence of bit flips from $x$, where each string in the sequence is a satisfying assignment to $\phi$. Let $T$ denote the span of all remaining satisfying assignments. Note that $\ket{x} \in S$, $\ket{y} \in T$. Also, the Hamming distance from any computational basis state in $S$ to computational basis state in $T$ is at least $2$; thus, $S$ and $T$ are $1$-orthogonal subspaces. From the Traversal Lemma (Lemma~\ref{l:traversal}), we know for any sequence of one-qubit unitaries $(U_i)_{i=1}^m$ that either $\norm{\ket{\psi_m} - \ket{\phi}}_2 \geq \ddd = 1/4$, or there exists an $i\in[m]$ such that $\bra{\psi_i} P' \ket{\psi_i} \geq ( {1}/(4m) )^2 = 2^{-(2n+4)}$,
where we again define ${\ket{\psi_i}:=U_i\cdots U_1\ket{\psi}}$ and $P' = I - \Pi_S - \Pi_T$. Thus, if it were the case that $H \succeq P'$, then
\[
	\bra{\psi_i} H \ket{\psi_i} \geq \bra{\psi_i} P' \ket{\psi_i} \geq \dfrac{1}{2^{2n+4}} = \eta_2,
\]
as desired. To see that indeed $H \succeq P'$, note that $H$ and $P'$ are diagonal matrices with non-negative integer entries satisfying for $z \in \{ 0, 1\}^n$:
\[
(\bra{z} H \ket{z} = 0 \iff \bra{z} P' \ket{z} = 0) \quad \text{ and } \quad
(\bra{z} H \ket{z} \geq 1 \iff \bra{z} P' \ket{z} = 1).
\]
This concludes the proof.

\end{proof}

\noindent Remark: Note that in \stcon, one may require exponentially many bit flips (i.e. exponential $m$) in general to map $x$ to $y$, as each bit flip must preserve the property that the current assignment is a satisfying assignment to the 3-CNF $\phi$. Thus, the Hamming distance between $x$ and $y$ is in general a loose lower bound on the number of bit flips required.

\subsection{Containment in \class{PSPACE}}

We now show that \gscon~is in \class{PSPACE} for exponentially many $1$-local unitaries $U_i$ and inverse exponential gap $\Delta$.

\begin{lemma}\label{l:inPSPACE}
    For all nonnegative constants $c_1$ and $c_2$, \gscon~with $l=1$ is in \class{PSPACE} for $m\leq2^{n^{c_1}}$ and $\Delta\geq 1/2^{n^{c_2}}$, where $n$ denotes the number of qubits $H$ acts on.
\end{lemma}
\begin{proof}
    We give a non-deterministic polynomial space algorithm for \gscon, and subsequently apply Savitch's theorem~\cite{S70} to obtain a \class{PSPACE} algorithm. Specifically, given a ${\gscon}$ instance ${\Pi=(H,\aaa,\bbb,\ccc,\ddd, \Delta, l, m, \ket{\psi},\ket{\phi})}$, our non-deterministic algorithm proceeds as follows. Let $L$ denote the number of local terms in $H$, and let $N$ denote the $\epsilon$-net for single qubit unitaries from Lemma~\ref{l:net} for $\epsilon:=\Delta/8L(2(m-1)+1)$. Then our algorithm is given by (explanation to follow):
\begin{algorithm}
 \begin{enumerate}
    \item If $\enorm{{\ket{\psi}}-{\ket{\phi}}}\leq \ccc$, accept.
    \item For $i\in\set{0,\ldots, m}$, define
$V_i := V_{i,1}\otimes\cdots \otimes V_{i,n}$ (for operators $V_{i,j}$ to be defined in iteration $i$).
    \item Set $V_{0,j} := I$ for all $j\in[n]$, i.e., $V_0 := I$.
    \item For $i = 1$ to $m$, do:
    \begin{enumerate}
        \item Non-deterministically guess a unitary $B_i\in\unitary{\complex^2}$ from $N$, where $B_i$ acts on some qubit $q\in[n]$ chosen non-deterministically.
        \item For $j\neq q$, set $V_{ij} := V_{i-1,j}$. Set $V_{iq}':=B_iV_{i-1,q}$.
        \item Set $V_{iq} := \round(V_{iq}')$, where $\round(A)$ straightforwardly maps unitary $A$ to a net element $A'\in N$ such that $\snorm{A-A'}\leq \epsilon$.
        \item (Energy Test) If $\bra{\psi}V_i^\dagger HV_i\ket{\psi}\geq \aaa + \Delta/3$, exit loop.
        \item (Proximity Test) If $\enorm{{V_i\ket{\psi}}-{\ket{\phi}}}\leq \ccc+ \Delta/4$, accept.
    \end{enumerate}
    \item Reject.
\end{enumerate}
\caption{Polynomial space algorithm for \gscon}
\label{alg:pspace}
\end{algorithm}

The intuition behind the algorithm is as follows. Ideally, we would like to run the following algorithm: At each step, non-deterministically guess a unitary $U\in\unitary{\complex^d}$, apply $U$ to the state computed in the previous step, and check whether the new state has high energy (Step 4(d)), or is close to the target state (Step 4(e)). Note that at a high level, this is possible in \class{PSPACE} because each unitary acts on a single qubit; thus, it suffices to keep track of the \emph{cumulative} single-qubit unitary applied to each qubit after each step (Step 4(b)), as opposed to keeping a history of all $m$ (i.e. exponentially many) unitaries guessed in Step 4. In particular, this implies the overall unitary $V_i$ in each iteration has a \emph{succinct} description (i.e., of tensor product form). There are, however, two subtle issues with this approach. The first is that the space of unitaries is continuous; thus, in iteration $i$, our algorithm non-deterministically chooses a unitary $B_i$ from $N$ instead (Step 4(a)). The second issue is that $m$ is exponentially large --- thus, multiplying all $B_i$ which act on a qubit $j$ can result in an operator whose entries require an exponential number of bits of precision. To prevent this, in each iteration, Step 4(c) ``rounds'' the product $B_iV_{i-1,q}$ back to an operator in our net. For completeness, note that Step 4(d) can be implemented using Kitaev's phase estimation procedure for placing the $k$-local Hamiltonian problem in QMA~\cite{K99}, and Step 4(e) can be implemented using the SWAP test~\cite{BCWW01}.

We now justify why the algorithm runs in polynomial space. Since each $V_i$ can be described using a polynomial number of bits, Step 4(a) can be carried out by a Turing machine whose configurations each require at most polynomially many bits to specify. For Step 4(c), since $\epsilon$ is inverse exponential in our setting, Lemma~\ref{l:net} implies $\abs{N}$ scales exponentially; thus, Step 4(c) can be achieved in polynomial space via a brute force search over all indices $i$ of operators in the net via Lemma~\ref{l:net}. Steps 4(d) and 4(e) can be completed in polynomial space using the standard approach of recomputing any values needed on-the-fly when determining (say) an inner product of exponentially large vectors specified by polynomial-size quantum circuits. We conclude that the algorithm runs in polynomial space.

We now justify correctness. Suppose first that there exists a sequence of $1$-local unitaries $(\hat{U}_{i})_{i=1}^{m}$ satisfying the conditions of a YES instance of \gscon. For convenience, define the global unitary after step $i$ as ${U_i:=U_{i,1}\otimes\cdots\otimes U_{i,n}}$. We prove by induction on $i$ that for all $i\in[m]$,
\begin{equation} \label{eqn:goal}
    \snorm{U_i-V_i} \leq (2(i-1)+1)\epsilon.
\end{equation}
For the base case $i=1$, we have $\snorm{U_1-V_1} \leq \epsilon$ since $\snorm{A\otimes B}=\snorm{A}\snorm{B}$ (recall $U_1$ and $V_1$ act non-trivially only on a single qubit) and by Lemma~\ref{l:net}. Thus, the base case holds. For the inductive step, assume the claim is true for iterations $1$ through $i-1$. We prove it for iteration $i$. Specifically,
\begin{eqnarray*}
    \snorm{U_i-V_i}&=&\snorm{\hat{U}_iU_{i-1}-\round(B_i V_{i-1})}\\
    &\leq&\snorm{\hat{U}_iU_{i-1}-B_i V_{i-1}}+ \snorm{B_i V_{i-1}-\round(B_i V_{i-1})}\\
    &\leq&\snorm{\hat{U}_i-B_i}+\snorm{U_{i-1}-V_{i-1}}+\epsilon\\
    &\leq&\epsilon + (2(i-2)+1)\epsilon + \epsilon\\
    &=&(2(i-1)+1)\epsilon,
\end{eqnarray*}
where the first inequality follows from the triangle inequality, the second inequality from the fact that $\snorm{AB-CD}\leq \snorm{A-C}+\snorm{B-D}$ for unitaries $A,B,C,D$ and by Lemma~\ref{l:net}, and the third inequality from Lemma~\ref{l:net} and the induction hypothesis. This completes our proof of Equation~(\ref{eqn:goal}).

We conclude that in any iteration $i\in[m]$, we have $\snorm{U_i-V_i}\leq (2(i-1)+1)\epsilon$, and hence $\enorm{U_i\ket{\psi}-V_i\ket{\psi}}\leq (2(i-1)+1)\epsilon$. Recalling that $L$ is the number of local terms in $H$, this yields
\begin{eqnarray}\label{eqn:est1}
    \abs{\bra{\psi}U_i^\dagger H U_i\ket{\psi}-\bra{\psi} V_i^\dagger H V_i\ket{\psi}}&\leq& \snorm{H}\trnorm{U_i\ketbra{\psi}{\psi}U_i^\dagger-V_i\ketbra{\psi}{\psi}V_i^\dagger}\nonumber\\&\leq& 2L\enorm{U_i\ket{\psi}-V_i\ket{\psi}}\nonumber\\&\leq& 2L(2(i-1)+1)\epsilon\nonumber\\
    &\leq&\frac{\Delta}{4},
\end{eqnarray}
where the first inequality follows from H\"{o}lder's inequality, and the second from Equation~(\ref{eqn:enorm}). In addition, since in a YES instance $\enorm{U_m\ket{\psi}-\ket{\phi}}\leq \ccc$, by the triangle inequality we have
\begin{equation}\label{eqn:est2}
    \enorm{V_{{m}} \ket{\psi}-\ket{\phi}}\leq \enorm{V_{{m}}\ket{\psi}-U_{{m}}\ket{\psi}}+\enorm{U_{{m}}\ket{\psi}-\ket{\phi}}\leq (2({{m}}-1)+1)\epsilon+\ccc\leq\frac{\Delta}{4}+\ccc.
\end{equation}
By Equations~(\ref{eqn:est1}) and~(\ref{eqn:est2}), we conclude that for a YES instance of \gscon, Step 4(d) of our algorithm will never cause an exit from the loop, and Step 4(e) will accept in some iteration. An analogous argument shows that for any NO instance, either the algorithm exits the loop in Step 4(d) or never passes the check in Step 4(e), implying the algorithm rejects, as desired.
\end{proof}

\section{\class{NEXP}-completeness} \label{sect:GSCON_nexp}

In this section, we define a succinct version of \gscon, and show that it is \class{NEXP}-complete. As the proof techniques used here are essentially the same as in Sections~\ref{sect:GSCON_QCMA} (\class{QCMA}-completeness) and~\ref{sect:GSCON_pspace} (\class{PSPACE}-completeness), for brevity we give only proof sketches.

We begin by defining succinct or \emph{oracle} notions of a local Hamiltonian and quantum product states, in analogy with an oracle \textup{$3$-CNF} formula and oracle truth assignment~\cite{BR04}.

\begin{definition}[Oracle $k$-local Hamiltonian]
Let $H=\sum_{i=1}^{2^r} H_i$ be a $k$-local Hamiltonian acting on $2^n$ qubits with $2^r$ clauses. An \emph{oracle} local Hamiltonian is a classical circuit $C_H$ which, given index $i\in\set{0,1}^r$ as input, outputs a classical description of constraint $H_i$ (i.e. outputs a $2^k\times 2^k$ matrix), along with the indices of the $k$ qubits on which $H_i$ acts.
\end{definition}

\begin{definition}[Oracle quantum product state]
Let $\ket{\psi}$ be a tensor product state on $2^n$ qubits such that $\ket{\psi}= \ket{\psi_1}\otimes\ket{\psi_2}\otimes\cdots\otimes\ket{\psi_{2^n}}$. An \emph{oracle} quantum product state is a classical circuit $C_\psi$ which, given index $i\in\set{0,1}^n$ as input, outputs a classical description of $\ket{\psi_i}$.
\end{definition}

Using these two definitions, we can now define the succinct version of \gscon.

\begin{definition}[\succgscon $({H},{k},{\aaa},{\bbb},{\ccc},{\ddd},{\Delta},{l},{m},{U_\psi},{U_\phi})$]
\emph{\succgscon} is defined identically to $\gscon$, except the Hamiltonian $H$ is an oracle Hamiltonian and the initial states $\ket{\psi}$ and $\ket{\phi}$ are oracle quantum product states.
\end{definition}

In this section, we show the following theorem.
\begin{theorem}\label{thm:nexpcomplete}
    \succgscon~is \class{NEXP}-complete for $m \in O(2^n)$, $\Delta\in{\Theta} (1/m^2)$, $l=1$, and $k\geq 5$, where $2^n$ is the number of qubits $H$ acts on.
\end{theorem}

\noindent Intuitively, this says that the succinct version of \gscon~in which (1) the number of unitaries scales linearly in the number of qubits (but exponentially in the input size) and (2) each unitary is $1$-local is \class{NEXP}-complete.

\subsection{\class{NEXP}-hardness}

We now show \class{NEXP}-hardness of \succgscon.

\begin{lemma}\label{l:nexphard}
    \succgscon~is \class{NEXP}-hard for $m \in O(2^n)$, $\Delta \in {O} (1/m^2)$, $l=1$, and $k \geq 5$, where $2^n$ is the number of qubits $H$ acts on.
\end{lemma}

\begin{proof}
We sketch a polynomial-time many-one or Karp reduction from the \class{NEXP}-complete problem \osat~(see, e.g.~\cite{BR04}) to \succgscon. Specifically, in an \osat~instance, one is given as input an oracle \textup{$3$-CNF} formula $\eta$ consisting of $2^n$ variables and $2^r$ clauses; $\eta$ can be thought of as a circuit $C_\eta$ which, given index $i\in\set{0,1}^m$, outputs the $i$'th clause and the indices of the variables on which the $i$'th clause acts.

Our approach is as follows: We embed the oracle \textup{$3$-CNF} formula into an oracle $3$-local Hamiltonian in the trivial way, and subsequently combine this with the construction of Lemma~\ref{l:qcmahard} (\class{QCMA}-hardness). Specifically, our oracle Hamiltonian $C_H$ acts as follows: Given index $i$, it runs $C_\eta$ on $i$ to obtain the $i$'th clause $c_i$. It then converts this to a diagonal Hamiltonian constraint $H_i$ (for example, clause $x_1\vee x_2\vee x_3$ is mapped to the diagonal operator $\operatorname{Diag}(1,0,0,0,0,0,0,0)$), and returns constraint $H_i\otimes P$ for $P:=I-\ketbra{00}{00}-\ketbra{11}{11}$. (Note that here $P$ needs only act on $2$ qubits since the unitaries $U_i$ are $1$-local.) The initial and final states are oracle quantum product states $C_\psi$ and $C_\phi$ representing $\ket{\psi}:=\ket{0}^{\otimes 2^n} \ket{0}^{\otimes 2}$ and $\ket{\phi}:=\ket{0}^{\otimes 2^n} \ket{1}^{\otimes 2}$, respectively. (Clearly, $C_\psi$ and $C_\phi$ have size $\poly(n)$.) Set $\aaa:=0$, ${\bbb:=1/(16m^2)}$, $\ccc:=0$, ${\ddd:=1/4}$, $\Delta\in {O}(1/m^2)$, and $l=1$. This concludes the construction of our $\succgscon$ instance.

To show correctness, for a YES instance, we proceed analogously to Lemma~\ref{l:qcmahard}, except now there is no history state to prepare; in particular, the sequence of $m$ unitaries is given by:
    \begin{enumerate}
        \item Apply Pauli $X$ gates to $h_1$ to prepare satisfying assignment $x$ for $\eta$, i.e. map $\ket{0}^{\otimes 2^n}$ to $\ket{x}$.
        \item Apply $(X\otimes I)_{G}$ to ``initiate'' checking of $\ket{x}$.
        \item Apply $(I\otimes X)_{G}$ to ``complete'' checking of $\ket{x}$.
        \item Apply $X$ gates to $h_1$ to map the initial proof $\ket{x}$ back to $\ket{0}^{\otimes 2^n}$.
    \end{enumerate}
Clearly, this process requires at most $m=2^{n+1}+2$ single-qubit unitaries, as desired. The analysis for a NO instance  proceeds essentially identically to Lemma~\ref{l:qcmahard}; one need only replace $\beta$ by $1$. The reason this works is because $H := \sum_i H_i \succeq I$ since it is a sum of diagonal projections and there does not exist a classical string $z$ such that $\bra{z} H \ket{z} = 0$.
\end{proof}

\subsection{Containment in \class{NEXP}}

We now show containment of \succgscon~in \class{NEXP}.

\begin{lemma}\label{l:inNEXP}
    \succgscon~with $l=1$ is in \class{NEXP} for $m\leq \poly(2^n)$ and $\Delta\geq 1/\poly(2^n)$, where $2^n$ is the number of qubits $H$ acts on.
\end{lemma}

\begin{proof}
    The proof is essentially identical to that of Lemma~\ref{l:inPSPACE} (containment in \class{PSPACE}), i.e. the verifier runs Algorithm~\ref{alg:pspace}. As the Hamiltonian involved now acts on exponentially many qubits, a few remarks regarding the implementation of Algorithm~\ref{alg:pspace} are in order:
    \begin{itemize}
        \item The initial state $\ket{\psi}$, final state $\ket{\phi}$, and intermediate states $V_i\ket{\psi}$ are product states. Hence, the Energy Test (Step 4(d)) and Proximity Test (Step 4(e)) can be carried out in exponential time. For example, suppose for the former that we wish to estimate $\bra{\psi}V_i^\dagger HV_i\ket{\psi}$. For this, it suffices to  estimate each $\bra{\psi}V_i^\dagger H_jV_i\ket{\psi}$ individually. If $H_j$ acts on qubits $q_1,q_2,q_3$, then we simply query $C_\psi$ for the original state of qubits $q_1,q_2,q_3$, apply $V_{i,q_1}\otimes V_{i,q_2}\otimes V_{i,q_3}$ to these three qubits, and finally compute the desired expectation against $H_j$.

        \item The verification procedure now requires exponential space, since we must keep track of exponentially many cumulative $1$-qubit operators $V_{iq}$ which comprise the global $i$'th operator $V_i = V_{i,1}\otimes\cdots \otimes V_{i,{2^n}}$.
    \end{itemize}
\end{proof}

\section{{On the} tightness of the Traversal Lemma and properties of $k$-orthogonality}\label{sect:tightprops}

{In the next two subsections, we discuss tightness of the Traversal Lemma and study the properties of $k$-orthogonality further.}

\subsection{On the tightness of the Traversal Lemma}\label{sect:tightness}

We now ask whether the Traversal Lemma is tight in the following sense: In Lemma~\ref{l:traversal}, the lower bound on $\bra{v_i}P\ket{v_i}$ scales as $\Theta(1/m^2)$ (for $m$ the number of unitaries and for fixed $\epsilon$). This intuitively suggests that one can better ``avoid'' the subspace $P$ projects onto if one uses a longer sequence of local unitaries. Is such behavior possible? Or can the lower bound in Lemma~\ref{l:traversal} be improved to a constant independent of $m$? In this section, we show that
a dependence on $m$ in Lemma~\ref{l:traversal} is indeed necessary.

\begin{theorem} \label{thm:badtraversal}
We assume the notation of Lemma~\ref{l:traversal}. Fix any $0<\Delta<1/2$, and consider $2$-orthogonal states $\ket{v}=\ket{000}$ and $\ket{w}=\ket{111}$, with $P:=I-\ketbra{v}{v}-\ketbra{w}{w}$. Then, there exists a sequence of $m$ $2$-local unitary operations mapping $\ket{v}$ to $\ket{w}$ through intermediate states $\ket{v_i}$, each of which satisfy $\bra{v_i}P\ket{v_i}\leq \Delta$, and where $m\in O(1/\Delta^2)$.
\end{theorem}

The idea behind the proof is based on the following rough analogy: Suppose one wishes to map the point $(1,1)$ (corresponding to $\ket{000}$) in the 2D Euclidean plane to $(-1,-1)$ (corresponding to $\ket{111}$) via a sequence of moves with the following two restrictions: (1) For each current point $(x,y)$, the next move must leave precisely one of $x$ or $y$ invariant (analogous to $2$-local unitaries acting on a $3$-qubit state), and (2) the Euclidean distance between $(x,y)$ and the line through $(1,1)$ and $(-1,-1)$ never exceeds $\Delta$ (analogous to the overlap with $P$ not exceeding $\Delta$). In other words, we wish to stay close to a diagonal line while making only horizontal and vertical moves. This can be achieved by making a sequence of ``small'' moves resembling a ``staircase''. The smaller the size of each ``step'' in the staircase, the better we approximate the line, at the expense of requiring more moves (analogous to increasing the number of unitaries, $m$). Although the idea in this analogy is appealing in its simplicity, applying it to the setting of the Traversal Lemma is non-trivial, requiring a careful selection of $2$-local unitary operations.

\begin{proof}[Proof of Theorem~\ref{thm:badtraversal}]
Our high level approach is as follows. We first give a unitary $U$ which is a sequence of two-qubit unitaries mapping $\ket{000}$ to $(\ket{000}+\ket{111})/\sqrt{2}$. Given the technique behind $U$'s construction, one can analogously obtain a unitary $V$  which maps $(\ket{000}+\ket{111})/\sqrt{2}$ to $\ket{111}$. It follows that $VU\ket{000}=\ket{111}$. It thus suffices to describe $U$, which is done in two steps. The first step consists of a pair of unitaries which transfer a small amount of amplitude from $\ket{000}$ to $\ket{111}$; applying this step repeatedly yields a state $\ket{\psi}$ ``close'' to $(\ket{000}+\ket{111})/\sqrt{2}$. It is this iterative repetition which causes the overall number of unitaries $m$ to scale as $\Omega(1/\Delta)$. Step 2 then maps $\ket{\psi}$ precisely onto $(\ket{000}+\ket{111})/\sqrt{2}$. We now describe these steps.

\paragraph{Step 1: Iteratively make small steps towards $(\ket{000}+\ket{111})/2$.} Given any state of the form $\gamma_1\ket{000}+\gamma_2\ket{111}$ for real $\gamma_1>\gamma_2$, we give a pair of two-qubit unitaries $(U_1,U_2)$ which intuitively transfer a small amount of amplitude from $\ket{000}$ to $\ket{111}$.

We first apply a two-qubit unitary $U_1$ to qubits $1$ and $2$ with action $\ket{00}\mapsto\alpha\ket{00}+\beta\ket{11}$ and $\ket{11}\mapsto\beta\ket{00}-\alpha\ket{11}$ (for real $\alpha,\beta$ to be specified later), obtaining:
\begin{eqnarray}
    \gamma_1\ket{000}+\gamma_2\ket{111}&\mapsto&\gamma_1\alpha\ket{000}+\gamma_1\beta\ket{110}+\gamma_2\beta\ket{001}-\gamma_2\alpha\ket{111}\nonumber\\
    &=&\ket{0}(\gamma_1\alpha\ket{00}+\gamma_2\beta\ket{01})+\ket{1}(\gamma_1\beta\ket{10}-\gamma_2\alpha\ket{11}).\label{eqn:trick}
\end{eqnarray}
The overlap of this state with $P$ is $\beta^2$.

Next, apply a unitary $U_2$ on qubits $2$ and $3$ with action (omitting normalization {for clarity})
$    (\gamma_1\alpha\ket{00}+\gamma_2\beta\ket{01}) \mapsto \ket{00}$ and
$    (\gamma_1\beta\ket{10}-\gamma_2\alpha\ket{11}) \mapsto \ket{11}$, obtaining:
\[
    \sqrt{\gamma_1^2\alpha^2+\gamma_2^2\beta^2}\ket{000}+\sqrt{\gamma_1^2\beta^2+\gamma_2^2\alpha^2}\ket{111}
\]
{which has $0$ overlap with $P$.}
Setting $\beta=\sqrt{\Delta}$ ensures this process has at most $\Delta$ overlap with $P$ at each intermediate step.

Let us now analyze the rate at which amplitude is transferred from $\ket{000}$ to $\ket{111}$ by this mapping. To do so, define
\[
{    f(\gamma_1):= \sqrt{\gamma_1^2(1-\Delta)+(1-\gamma_1^2)\Delta} = \sqrt{(1-2\Delta)\gamma_1^2+\Delta},}
\]
which is the {new amplitude after} the map induced by $U_2U_1$ is applied to input amplitude $\gamma_1$. {Note that $f(\gamma_1) \geq 0$ for all $\gamma_1$ and $f(\gamma_1)^2 > 1/2$ when $\gamma_1^2 > 1/2$.}


{We now {quantify} how \emph{much} amplitude is transferred from $\ket{000}$ to $\ket{111}$ by this process}. Suppose we iteratively apply $U_2U_1$ so long as $\gamma_1^2\geq1/2+\zeta$ for some cutoff ${\zeta \in (0, 1/2)}$. Then, the difference $\gamma_1^2-f(\gamma_1)^2$ satisfies
\[
    2\Delta\zeta\leq\gamma_1^2-f(\gamma_1)^2=\Delta(2\gamma_1^2-1)\leq \Delta .
\]
In other words, each iterative step moves us at least $2\Delta\zeta$ and at most $\Delta$ towards our cutoff ${1/2+\zeta}$, implying the number of iterations required to reach the cutoff scales between $\Omega(1/\Delta)$ and $O(1/\Delta\zeta)$. Setting $\zeta := 2\Delta/(1+2\Delta)$ (which lies in $(0,1/2)$ when $\Delta \in (0, 1/2)$)  ensures that the number of iterations is at most $O(1/\Delta^2)$, and sets up the next step of our transformation, which we now discuss.


\paragraph{Step 2: Map an ``almost'' equal superposition to an equal superposition.} After $O(1/\Delta^2)$ iterations of Step $1$, we arrive at a state of the form $\gamma_1\ket{000}+\gamma_2\ket{111}$ where {$\gamma_1 \geq 0$ satisfies}
\begin{equation}\label{eqn:step2_1}
    \frac{1}{2}<\gamma_1^2\leq \frac{1}{2} + \frac{2\Delta}{1+2\Delta}.
\end{equation}
We seek a sequence of one and two-qubit unitaries which map this state to ${(\ket{000}+\ket{111})/\sqrt{2}}$. To attain this, we instead equivalently give a sequence of unitaries which {achieves} the reverse mapping.

To begin, {we apply a unitary to qubits} $1$ and $2$ with {action} $\ket{00}\mapsto\ket{00}$ and $\ket{11}\mapsto(\beta\ket{0}+\alpha\ket{1})\ket{1}$ for real {parameters} $\alpha,\beta$ to be {specified later}. {This maps} $(\ket{000}+\ket{111})/\sqrt{2}$ to
\begin{equation}\label{eqn:step1_1}
    \frac{1}{\sqrt{2}}\ket{000}+\frac{\beta}{\sqrt{2}}\ket{011}+\frac{\alpha}{\sqrt{2}}\ket{111}=\delta\ket{0}\left(\frac{1}{\sqrt{2}\delta}\ket{00}+\frac{\beta}{\sqrt{2}\delta}\ket{11}\right)+\frac{\alpha}{\sqrt{2}}\ket{1}\ket{11},
\end{equation}
{where} $\delta:=\sqrt{(1+\beta^2)/2}$. The overlap with $P$ at this point is $\beta^2/2$.

Next, define a unitary with action:
\[
    \frac{1}{\sqrt{2}\delta}\ket{00}+\frac{\beta}{\sqrt{2}\delta}\ket{11}\mapsto \ket{00}, \quad
    \frac{\beta}{\sqrt{2}\delta}\ket{00}-\frac{1}{\sqrt{2}\delta}\ket{11} \mapsto \ket{11}, \quad
    \ket{01}\mapsto\ket{01}, \quad
    \ket{10}\mapsto\ket{10}.
\]
Since {this unitary} is Hermitian, it follows that $\ket{11}$ {is mapped to} $\frac{\beta}{\sqrt{2}\delta}\ket{00}-\frac{1}{\sqrt{2}\delta}\ket{11}$; hence, applying {this unitary} to qubits $2$ and $3$ in Eqn.~(\ref{eqn:step1_1}) yields:
\[
    \delta\ket{000}+\frac{\alpha\beta}{2\delta}\ket{100}-\frac{\alpha}{2\delta}\ket{111}=\left(\delta\ket{00}+\frac{\alpha\beta}{2\delta}\ket{10}\right)\ket{0}-\frac{\alpha}{2\delta}\ket{11}\ket{1}.
\]
This state has overlap $\alpha^2\beta^2/(2\delta)^2\leq \beta^2/2$ with $P$.

Finally, apply a unitary on qubits $1$ and $2$ which maps $\ket{11}$ to ${-\ket{11}}$ and (the normalized version of) $\delta\ket{00}+\frac{\alpha\beta}{2\delta}\ket{10}$ to $\ket{00}$, obtaining
\begin{equation}\label{eqn:step1_3}
    \sqrt{1-\frac{\alpha^2}{4\delta^2}}\ket{000}+\frac{\alpha}{2\delta}\ket{111}.
\end{equation}

It remains to set $\beta$ so as (1) to prevent overlap more than $\Delta$ with $P$, i.e. we require $\beta^2/2\leq \Delta$, and (2) to ensure that the amplitude on $\ket{000}$ in Eqn.~(\ref{eqn:step1_3}) is precisely $\gamma_1$. Defining $\beta$ implicitly via the equation
\[
    \sqrt{1-\frac{\alpha^2}{4\delta^2}}=\sqrt{1-\frac{1-\beta^2}{2(1+\beta^2)}}=\gamma_1
\]
clearly satisfies the second of these requirements. Using the upper bound on $\gamma_1^2$ from Eqn.~(\ref{eqn:step2_1}), it is straightforward to verify that the first requirement is also met.
\end{proof}

\subsection{Properties of $k$-orthogonality}\label{app:prop}

{We now study the properties of $k$-orthogonality further, and give an intuitive characterization of the notion (Lemma \ref{l:prop2}).} We hope this may prove useful in possible independent applications of the concepts introduced in this {work}.

 We begin with the following useful lemma.

 \begin{lemma}\label{l:prop1}
     For any $\ket{v},\ket{w}\in(\complex^d)^{\otimes n}$, $\ket{v}$ and $\ket{w}$ are $k$-orthogonal if and only if for all subsets of qudits $S\subseteq[n]$ of size at most $k$, we have $\trace_{[n]\setminus S}(\ket{v}\bra{w})=0$.
 \end{lemma}

 \begin{proof}
     Assume first that $\ket{v}$ and $\ket{w}$ are $k$-orthogonal, and consider any $S\subseteq [n]$ with $\abs{S}\leq k$. Then, we have
     \[
         0
         =
         \max_{U_S\in \unitary{(\complex^d)^{\otimes \abs{S}}}}\abs{\bra{w}I\otimes U_S\ket{v}}
         =
         \max_{U_S\in \unitary{(\complex^d)^{\otimes \abs{S}}}}\abs{\inner{U_S}{\trace_{[n]\setminus S}(\ket{v}\bra{w})}}
         =
         \trnorm{\trace_{[n]\setminus S}(\ket{v}\bra{w})},
     \]
     where the second equality follows since $\trace((I_A\otimes C_B) D_{AB})=\trace(C\trace_A(D_{AB}))$ for all linear operators $C$ and $D$, and the third equality since $\trnorm{A}=\max_{U\in\unitary{\spa{X}}}\abs{\trace(UA)}$~\cite{W08_2} (intuitively, this holds since the optimal $U$ rotates the set of left singular vectors of $A$ into the set of right singular vectors of $A$). But now the claim follows, since $\trnorm{A}=0$ if and only if $A=0$. The converse direction proceeds analogously.
 \end{proof}


 \begin{lemma}\label{l:prop2}
     For any $\ket{v},\ket{w}\in(\complex^d)^{\otimes n}$, $\ket{v}$ and $\ket{w}$ are $k$-orthogonal if and only if for all subsets of qudits $S\subseteq[n]$ of size at most $k$, we have $(\trace_{S}\ketbra{v}{v})(\trace_{S}\ketbra{w}{w})=0$.
 \end{lemma}

 \begin{proof}
     Assume first that $\ket{v}$ and $\ket{w}$ are $k$-orthogonal, and consider any $S\subseteq [n]$ with $\abs{S}\leq k$. Let $Y$ denote the register corresponding to $[n]\setminus S$. Then, suppose the Schmidt decompositions of $\ket{v}$ and $\ket{w}$  are $\ket{v}=\sum_i\alpha_i\ket{a_i}_S\ket{b_i}_Y$ and $\ket{w}=\sum_j\beta_j\ket{c_j}_S\ket{d_j}_Y$, respectively. Now, by Lemma~\ref{l:prop1}, we have
     \[
         0 = \trace_{Y}(\ket{v}\bra{w})=\sum_{ij}\alpha_i\beta_j\brakett{d_j}{b_i}\ketbra{a_i}{c_j}.
     \]
     Since ${\{ \, \ket{a_i} \}}$ and ${\{ \, \ket{c_j} \}}$ are orthonormal sets, this implies that for any $i$ and $j$, either $\alpha_i=0$, $\beta_j=0$, or $\brakett{d_j}{b_i}=0$. Thus, letting
 \[ \rho:=\trace_S(\ketbra{v}{v})=\sum_i\alpha_i^2\ketbra{b_i}{b_i}\quad \text{ and } \quad \sigma:=\trace_S(\ketbra{w}{w})=\sum_j\beta_j^2\ketbra{d_j}{d_j}, \]
 we have $\trace(\rho\sigma)=\sum_{ij}\alpha_i^2\beta_j^2\abs{\brakett{d_j}{b_i}}^2=0$, or equivalently, $\rho\sigma=0$ since $\rho,\sigma\succeq 0$. The converse direction proceeds analogously.
 \end{proof}

 Using Lemma~\ref{l:prop1}, we can also easily show the following statement regarding ``extensions'' of $k$-orthogonal states.

 \begin{lemma}
     For any $\ket{v},\ket{w}\in(\complex^d)^{\otimes n}$, $\ket{v}$ and $\ket{w}$ are $k$-orthogonal if and only if for \emph{all} $\ket{V},\ket{W}\in(\complex^{d})^{\otimes {n'}}$ for $n'\geq 0$, $\ket{v}\ket{V}$ and $\ket{w}\ket{W}$ are  $k$-orthogonal. (Note: What makes this not completely trivial is that the $k$-local unitary can act across the cut between $\ket{v}$ and $\ket{V}$.)
 \end{lemma}
 \begin{proof}
     Assume first that $\ket{v}$ and $\ket{w}$ are $k$-orthogonal, and consider arbitrary $n' \geq 1$ and vectors $\ket{V},\ket{W} \in(\complex^d)^{\otimes n'}$. Let $S \subseteq [n]$ and $S' \subseteq [n']$ be such that $|S \cup S'| \leq k$. Then we have
 \[
     \trace_{[n] \cup [n'] \setminus (S \cup S')}(\ket{v} \ket{V}\bra{w} \bra{W}) = \trace_{[n] \setminus S}(\ket{v} \bra{w}) \trace_{[n'] \setminus S'}(\ket{V} \bra{W})   = 0,
 \]
 where the last equality holds since $\ket{v}$ and $\ket{w}$ are $k$-orthogonal and by Lemma~\ref{l:prop1}. Thus, $\ket{v}\ket{V}$ and $\ket{w}\ket{W}$ are $k$-orthogonal. Since $\ket{V}$ and $\ket{W}$ are arbitrary, this direction of the claim holds. The converse statement is trivially true.
 \end{proof}

\section{Conclusions and open problems} \label{sect:conclusion}

In this paper, we defined a physically motivated notion of connectivity for ground {spaces} of quantum local Hamiltonians, and initiated its study. Specifically, we asked: Given a local Hamiltonian $H$ and initial and final states $\ket{\psi}$ and $\ket{\phi}$, respectively, can $\ket{\psi}$ be mapped via local unitary operations to $\ket{\phi}$ through the ground space (more generally, through the low-energy space) of $H$? Our main results showed that the complexity of this problem can range from \class{QCMA}-complete to \class{NEXP}-complete, depending on the specific formulation of the problem. As a result, we obtained a natural \class{QCMA}-complete problem, adding to the short list of known QCMA-complete problems. To show this \class{QCMA}-hardness result, we proved the Traversal Lemma, which allows one to analyze the path a unitary evolution must take in certain settings. We further showed that the Traversal Lemma is tight up to a polynomial factor in the length of the unitary evolution considered.

We close with the following open problems. (1) References~\cite{GKMP06} and~\cite{MNPR14} show dichotomy and trichotomy theorems, respectively, for classical reconfiguration problems involving Boolean satisfiability; can similar theorems be shown in the quantum setting? For example, are there non-trivial quantum cases of \gscon~which can be solved in P or BQP? (2) Our complexity theoretic results on GSCON depended crucially on the parameters $m$ (the number of unitaries) and $l$ (the locality of each unitary). We have shown that polynomial $m$ and $l=2$ characterizes QCMA, and that exponential $m$ and $l=1$ characterizes PSPACE. There is, however, an interesting regime left to consider: Exponential $m$ and $l=2$. In this case, our proof of containment in PSPACE seems to fail as each intermediate state in the evolution appears to require exponential space to represent. However, this variant of the problem is in NEXP, and we conjecture that it is in fact \class{NEXP}-complete. (3) Regarding our Traversal Lemma, can it (or some variant thereof) be used in other settings in quantum computational complexity, such as in analyzing quantum adiabatic algorithms? (4) Finally, are there other problems related to \gscon~which are also complete for quantum complexity classes such as \class{QCMA}?

\section*{Acknowledgements}

We thank Joel Klassen for suggesting the connection between \gscon~and quantum memories, {Barbara Terhal and Roberto Oliveira for helpful discussions regarding perturbation theory gadgets}, Sarvagya Upadhyay, Damian Markham, Eleni Diamanti, Attila Pereszlenyi, Amer Mouawad, Vinayak Pathak, and David Gosset for insightful discussions, {Shelby Kimmel for suggesting the use of perfect completeness of QCMA~\cite{JKNN12} to improve QCMA-completeness of \gscon\ to QCMA-completeness of \ffgscon,} and the anonymous referees whose feedback helped improve this paper. Part of this work was completed during the UC Berkeley Simons Institute for the Theory of Computing's program on Quantum Hamiltonian Complexity.

SG acknowledges support from a Government of Canada NSERC Banting Postdoctoral Fellowship and the Simons Institute for the Theory of Computing at UC Berkeley.

JS acknowledges support from a Government of Canada NSERC Postdoctoral Fellowship, the French National Research Agency (ANR-09-JCJC-0067-01), and the European Union (ERC project QCC 306537). Research at the Centre for Quantum Technologies at the National University of Singapore is partially funded by the Singapore Ministry of Education and the National Research Foundation, also through the Tier 3 Grant ``Random numbers from quantum processes,'' (MOE2012-T3-1-009).


\bibliographystyle{alpha}
\bibliography{Sevag_Gharibian_Central_Bibliography_Abbrv,Sevag_Gharibian_Central_Bibliography}

\appendix

\section{Proofs for Section~\ref{sect:nets}}\label{app:netproofs}

\begin{proof}[Proof of Lemma~\ref{l:net}.]
    Any $2\times 2$ unitary $U$ can be written in terms of parameters $0 \leq x \leq 1$ and $0 \leq \phi_1, \phi_2, \phi_3 \leq 2\pi$ such that
    \begin{equation}\label{eqn:2x2}
        U = \left(
            \begin{array}{cc}
              \sqrt{x} \, e^{i \phi_1} & \sqrt{1-x} \, e^{i \phi_2} \\
              \sqrt{1-x} \, e^{i \phi_3} & \sqrt{x} \, e^{i \phi_4} \\
            \end{array}
          \right),
    \end{equation}
where we let $\phi_4 := -\phi_1 + \phi_2 + \phi_3 + \pi$ for brevity. The net is constructed by a straightforward discretization of the ranges of $x$, $\phi_1$, $\phi_2$, and $\phi_3$ into segments of size $\delta>0$, for $\delta$ to be chosen as needed. For any unitary $U$, there hence exist parameters $0 \leq y \leq 1$ and $0 \leq \theta_1, \theta_2, \theta_3 \leq 2\pi$ in the discretization such that $|x-y|, |\phi_1 - \theta_1|, |\phi_2 - \theta_2|, |\phi_3 - \theta_3| \leq \delta$. We now upper bound $\snorm{U - \widetilde{U}}$, where we have defined the unitary matrix
    \[
        \widetilde{U} := \left(
            \begin{array}{cc}
              \sqrt{y} \, e^{i \theta_1} & \sqrt{1-y} \, e^{i \theta_2} \\
              \sqrt{1-y} \, e^{i \theta_3} & \sqrt{y} \, e^{i \theta_4} \\
            \end{array}
          \right)
    \]
with $\theta_4 := -\theta_1 + \theta_2 + \theta_3 + \pi$ (implying $|\phi_4 - \theta_4| \leq 3 \delta$).
We first upper bound the magnitude of each entry of $U - \widetilde{U}$ individually. For $j \in \set{1,4}$, we have
\begin{eqnarray*}
|\sqrt{x} \, e^{i \phi_j} - \sqrt{y} \, e^{i \theta_j}|
& \leq & |\sqrt{x} \, e^{i \phi_j} - \sqrt{x} \, e^{i \theta_j}| + |\sqrt{x} \, e^{i \theta_j} - \sqrt{y} \, e^{i \theta_j}| \\
& = & \sqrt{x} |e^{i \phi_j} - e^{i \theta_j}| + |\sqrt{x} - \sqrt{y}| \\
& \leq & |\phi_j - \theta_j| + \sqrt{|x - y|} \\
& \leq & 4 \sqrt{\delta},
\end{eqnarray*}
{where we used} $x \leq 1$,  $\abs{e^{i\phi}-e^{i\theta}}\leq\abs{\phi-\theta}$ when $\abs{\phi-\theta} \leq 1$, and the inequality $|\sqrt{x} - \sqrt{y}| \leq \sqrt{|x-y|}$. The same argument yields $|\sqrt{1-x} \, e^{i \phi_j} - \sqrt{1-y} \, e^{i \theta_j}| \leq 4 \sqrt{\delta}$ for $j \in \set{2, 3}$.

We now use our bounds on each entry of $U-\tilde{U}$ as follows. For $\mnorm{A}:=\max_{ij}\abs{A_{ij}}$, it holds that $\snorm{A}\leq d\mnorm{A}$ for any $A\in\lin{\complex^d}$~(see, e.g., \cite{HJ90}). Hence,
    \[
        \snorm{U - \widetilde{U}} \leq 8 \sqrt{\delta}.
    \]
    Thus, in order to obtain an $\epsilon$-net over single-qubit unitaries, it suffices to set $\delta=\epsilon^2/64$.

    To complete the proof of our claim, we now need to bound the size of our net. Since we have $4$ parameters $x, \phi_1, \phi_2, \phi_3$, each discretized into segments of length $\delta \in O(\epsilon^2)$, our net contains $O(\epsilon^{-8})$ elements. Ordering our net elements by canonically ordering the discretization of each individual parameter $x, \phi_1, \phi_2, \phi_3$ thus implies we can represent each $U_i$ in our net using $O(\log(1/\epsilon))$ bits and retrieve $U_i$ in time $O(\log^2(1/\epsilon))$.
\end{proof}

 \begin{proof}[Proof of Lemma~\ref{l:pseudonet}.]
     The construction of $N$ is straightforward: Cast a $\delta$-net over the unit disk for each entry $(i,j)$ of a $d\times d$ complex matrix, for $\delta$ to be chosen as needed. For the checking algorithm $C$, let $\ket{u_i}$ denote the $i$'th column of $\widetilde{U}\in N$. Then, defining $B:=\sum_{i=1}^d \ketbra{u_i}{u_i}$, $C$ accepts if and only if
     \begin{equation}\label{eqn:B}
         \snorm{B-I}\leq \frac{\epsilon}{2({d}+\epsilon)}.
     \end{equation}
    Finally, the rounding algorithm $R$ maps input $\widetilde{U}\in N$ to a matrix $U$ whose $i$'th column is given by $\ket{u_i'}:=B^{-1/2}\ket{u_i}$. We remark that the rounding algorithm is heavily inspired by the epsilon net construction in~\cite{PGACHW11,G13}.

     In order to proceed with the proof, we require a $\delta'$-net $D'$ over $d$-dimensional vectors, where $\delta':= \epsilon/[6d({d}+\epsilon)]$. For this, let $D$ denote our  $\delta$-net cast over the unit disk in our construction of $N$, and set $\delta := \delta'/\sqrt{d}$. Then, we claim that $D':=D^{\times d}$ gives us the desired $\delta'$-net over $\complex^d$. To see this, for $\ve{v}\in \complex^d$, let $\ve{w}$ be the vector obtained by snapping the coordinates of $\ve{v}$ to the $\delta$-net. Then,
      \[
          \enorm{\ve{v}-\ve{w}}= \sqrt{\sum_{i=1}^d (v_i - w_i)^2} \leq \sqrt{d \delta^2}=\delta'.
      \]

     We now prove that $N$ is an $\epsilon$-pseudo-net. Let $U\in\unitary{\complex^d}$. We first show that there exists $\widetilde{U}\in N$ such that $C$ accepts $\widetilde{U}$, and that $\inlinesnorm{U-\widetilde{U}}\leq \epsilon$. We proceed as follows: For each column $\ket{u_i}$ of $U$, replace it with a {$\delta'$-close} vector $\ket{u_i'}\in D'$. Letting $\widetilde{U}$ denote the resulting matrix, note that $\widetilde{U}\in N$. We now show the required two properties:

     \begin{enumerate}
     	\item ($\widetilde{U}$ is accepted by $C$) Let $A:=\sum_{i=1}^d\ketbra{u_i'}{u_i'}$. Then,
        \begin{equation}\label{eqn:Abound}
              \snorm{A-I}
               \leq \sum_{i=1}^d \snorm{\ketbra{u_i'}{u_i'}-\ketbra{u_i}{u_i}}
               \leq \sum_{i=1}^d \fnorm{\ketbra{u_i'}{u_i'}-\ketbra{u_i}{u_i}}
               \leq (2 + \delta') d \delta'
               \leq 3d \delta',
          \end{equation}
          where the last inequality follows since $\delta' \leq 1$, and the third inequality follows from Equation~(\ref{eqn:fnormjj}) and the fact that
          $	\norm{ \ket{u'_i} }_2 \leq \norm{ \ket{u_i}}_2 + \norm{\ket{u'_i} - \ket{u_i}}_2 \leq \delta' + 1.
          $
          Thus, $\widetilde{U}$ is accepted by $C$ since
      	\[
      		\snorm{A-I} \leq 3d \delta' =\frac{\epsilon}{2({d}+\epsilon)}.
      	\]

      \item ($\inlinesnorm{U-\widetilde{U}}\leq \epsilon$) We have
      \begin{equation}\label{eqn:simplerproof}
          \snorm{U-\widetilde{U}}
          \leq
          \sum_{i=1}^d \snorm{\ketbra{u_i}{i} - \ketbra{u_i'}{i}}
          =
          \sum_{i=1}^d \norm{\ket{u_i} - \ket{u_i'}}_2
          \leq d \delta'
          \leq \epsilon,
      \end{equation}
	where the second inequality holds since $D'$ is a $\delta'$-net.
     \end{enumerate}

     Conversely, suppose that $\widetilde{U}\in N$. We show that if $\widetilde{U}$ is accepted by $C$, then $R$ maps $\widetilde{U}$ to a unitary $U\in \unitary{\complex^d}$ such that $\inlinesnorm{\widetilde{U}-U}\leq \epsilon$. To do this, we first show that $B$ (as {used} in Equation~(\ref{eqn:B})) is invertible (otherwise, the algorithm $R$ we have described is not well-defined). Indeed, suppose to the contrary that $B\ket{v}=0$ for unit vector $\ket{v}$. Then, ${\enorm{(B-I)\ket{v}}=1}$. But this contradicts the fact that $C$ accepts $\widetilde{U}$, i.e., ${\snorm{B-I}\leq \epsilon/[2({d}+\epsilon)]<1}$. Next, observe that $U$ is unitary since
     \[
         \sum_{i=1}^d\ketbra{u_i'}{u_i'}=\sum_{i=1}^dB^{-1/2}\ketbra{u_i}{u_i}B^{-1/2}=B^{-1/2}BB^{-1/2}=I.
     \]
     Finally, to show that $\inlinesnorm{\widetilde{U}-U}\leq \epsilon$, by the same argument as in Equation~(\ref{eqn:simplerproof}), we have
     \begin{equation}\label{eqn:Ubound}
\snorm{U-\widetilde{U}}
          \leq
          \sum_{i=1}^d \norm{\ket{u_i} - \ket{u_i'}}_2
          =\sum_{i=1}^d \norm{(I-B^{-1/2})\ket{u_i}}_2
          \leq d \snorm{I-B^{-1/2}}.
     \end{equation}
     Thus, we are left to upper bound $\inlinesnorm{I-B^{-1/2}}$. We instead first upper bound $\inlinesnorm{I-B}$; using an argument analogous to Equation~(\ref{eqn:Abound}), we have that $\inlinesnorm{I-B}\leq 3d\delta'$. Applying now the fact that if $x\neq0$ and $\abs{x-1}\leq y$, then $|(1/\sqrt{x})-1|\leq y/(1-y)$, it follows that $\inlinesnorm{I-B^{-1/2}}\leq (3d\delta')/(1-3d\delta')$. Substituting this bound into Equation~(\ref{eqn:Ubound}), we conclude that $\inlinesnorm{U-\widetilde{U}}\leq \epsilon$. This completes the proof that $N$ constitutes an $\epsilon$-pseudo-net.

 Next, to bound the size of the net $N$, note that since $\delta\in\Theta(\epsilon/d^{5/2})$, a trivial construction of a $\delta$-net over the unit disk (i.e., place a square lattice of points down on the unit disk) has $O(d^5/\epsilon^2)$ elements. Since we cast the $\delta$-net over $d^2$ matrix entries, the size of $N$ is $O(d^{7}/\epsilon^2)$.

 Finally, to compute $\widetilde{U}_i$ given $i$ using $O(d^2\log^2(d^{5/2}/\epsilon))$ bit operations, note that $i$ encodes the entries of $d^2$ matrix positions $(s,t)$ of $\widetilde{U}_i$, each of which requires {$\log(d^{5/2}/\epsilon)$} bits\footnote{Simply encode the offsets on the imaginary and real axes.} to encode which element from the $\delta$-net we have at position $(s,t)$ . Since $\widetilde{U}_i$ has $d^2$ entries which need to be computed given $i$, the claim follows.
 \end{proof}


\end{document}